\documentclass[11pt]{article}
\usepackage[affil-it]{authblk}
\usepackage{color}
\usepackage{amsmath,amsfonts,amssymb,mathrsfs,amsthm}
\numberwithin{equation}{section}
\numberwithin{figure}{section}
\usepackage{graphicx}
\usepackage{listings}
\usepackage{enumerate,enumitem}
\usepackage{mathtools}
\usepackage[a4paper, total={6in, 8in}]{geometry}
\usepackage{cite}
\usepackage[colorlinks,linkcolor=blue,citecolor=blue]{hyperref}
\usepackage{empheq}
\usepackage{pdfpages}
\usepackage[titletoc]{appendix}
\usepackage[utf8]{inputenc}
\usepackage{xspace}
\usepackage{footmisc}

\newtheorem{defn}{Definition}[section]

\newtheorem{Thm}{Theorem}

\newcommand{\Figref}[1]{Fig.~\ref{#1}}

\newcommand{\C}{\mathcal{C}}
\newcommand{\K}{\mathcal{K}}
\newcommand{\Sectionref}[1]{Section~\ref{#1}}

\newcommand{\Eqref}[1]{Eq.~\eqref{#1}}
\newcommand{\Defref}[1]{Def.~\ref{#1}}
\newcommand{\Eqsref}[1]{Eqs.~\eqref{#1}}

\newcommand{\Thmref}[1]{Theorem~\ref{#1}}

\newcommand{\V}{\mathcal{V}}

\newcounter{mnotecount}[section]
\let\oldmarginpar\marginpar
\renewcommand\marginpar[1]{\-\oldmarginpar[\raggedleft\footnotesize #1]%
	{\raggedright\footnotesize #1}}

\title{Bianchi I ``asymptotically Kasner'' solutions of the Einstein scalar field equations}
\author{J. Ritchie\footnote{Email: jritchie@maths.otago.ac.nz}}
\affil{Department of Mathematics and Statistics, University of Otago, New Zealand.}
\begin{document}
	
	\maketitle
	
	\begin{abstract}
		In this work we investigate the asymptotic behaviour of solutions to the Einstein equations with a minimally coupled scalar field. The primary focus of the present paper here establishing under what conditions a solution becomes ``asymptotically Kasner'' sufficiently close to the initial singularity. To address this question we restrict our attention to Bianchi I space-times. By restricting our attention to a strictly monotonic scalar field we are able to provide necessary conditions on a potential so that the resulting solution is asymptotically Kasner. Moreover, we provide both explicit and numerical examples of asymptotically Kasner space-times.
	\end{abstract}
	
	\section{Introduction}
	The \emph{Kasner solutions} are a one-parameter family of spatially homogeneous cosmological solutions of the Einstein Field Equations (EFEs), without matter. These solutions are anisotropic, as each spatial axis is allowed to grow (or decay) at a different rate in time (with respect to a particular `natural' coordinate system) \cite{KasnerMetric}.

	These solutions play an important role in cosmology, in part due to the conjecture that spatially inhomogeneous solutions of the EFEs can be matched point-wise to a Kasner solution. In this setting spatial derivatives are believed to be negligible. Solutions of this nature are called \emph{asymptotically velocity term dominated (AVTD)}. However, an array of heuristic and numerical results have found that generic cosmological solutions are not AVTD in any gauge \cite{BGIMW,BIW,BM,BM2,BM3}. Instead, it is expected that generic solutions are oscillatory \cite{BIW,BM,BM2}. In this picture, it is conjectured that the solutions can be modelled point-wise by an anisotropic Kasner solution for some time interval (known as a \emph{Kasner epoch}), before jumping to a \emph{different} Kasner solution \cite{PhysRevD.69.063514}. This type of effect is known as ``mixmaster" behaviour and was described by Misner in \cite{MixMasterUniverse} and, independently, by Belinski-Khalatnikov and Lifshitz (BKL) in \cite{OscillatoryApproach:BKL} (see \cite{Heinzle_2009}, and the references therein, for an overview of mixmaster dynamics). The conjecture that generic cosmological solutions behave this way is known as the \emph{BKL conjecture}. The singular nature of the ``Big Bang" coupled to the (expected) oscillatory behaviour of the solutions makes studying this conjecture difficult, both analytically and numerically.

	In order to simplify the task of studying cosmological solutions near the Big Bang, it is common to couple the EFEs to a scalar field. In the spatially homogeneous setting this allows one to generalise the standard Kasner solutions to the \emph{Kasner scalar field solutions} (also commonly referred to as the `generalised Kasner solutions'). This is a two-parameter family of solutions, that are generically anisotropic. The only isotropic member of the Kasner scalar field solutions is the Friedman-Lema\^itre-Robertson-Walker universe (FLRW). Introducing a scalar field is advantageous as it has the effect of mollifying the expected oscillatory behaviour. In this setting there are reasons to believe that generic solutions of the Einstein scalar field equations \emph{are} AVTD. For example, in \cite{Andersson} Andersson and Rendall were able to show the existence of an infinite-dimensional family of solutions to the Einstein scalar field equations with AVTD asymptotics. These solutions are not limited to being close to a FLRW solution. In their work, the question of stability was not addressed. In the remarkable work by Rodanski and Speck \cite{RodnianskiSpeck:Linear,RodnianskiSpeck:NonLinear}, it was shown the AVTD behaviour was non-linearly stable in the sense that there is an open set around FLRW in which AVTD behaviour holds. The work by Rodanski and Speck is certainly impressive, but it does not identify the asymptotic degrees of freedom. This question was (partially) addressed by Beyer et al. for a linearised sub-system of the EFEs coupled to a scalar field \cite{LapseScalarField}.

	In each of the above works it is assumed that the scalar field potential is zero and as such it is unclear if their results still hold when a potential is included. This is not to say that it is uncommon to add a potential. On one hand, works such as \cite{PhysRevD.48.4662,PhysRevD.48.4669,Beyer_2013} consider evolutions \emph{away} from the initial singularity. In this setting, the potential is commonly used to address problems about inflation and graceful exit. On the other hand, works such as \cite{PhysRevD.61.023508,KasnerSolutions,Narita_2000,Soviet_Belinkskii,Berger_2004,Weaver_1999} consider the addition of a potential when evolving \emph{toward} the initial singularity. In \cite{PhysRevD.61.023508,KasnerSolutions} heuristic evidence is given that solutions are AVTD only if the potential decays appropriately. In works such as \cite{Narita_2000,Soviet_Belinkskii} it is noted that if the scalar field is coupled to the Maxwell equation then the mixmaster oscillations are restored. Moreover, \cite{PhysRevD.61.023508} claims that if the potential is of a particular exponential form, then it is also possible to restore the mixmaster oscillations.
	
	In this work we seek to answer the following questions: \emph{{Can a potential be introduced so that resulting solutions are asymptotically Kasner?}} If so, \emph{{what kind of potentials do not lead to asymptotically Kasner solutions?}} And \emph{{how would they differ from the standard Kasner scalar field solutions?}} Questions of this nature have been previously considered by Condeescu et al in \cite{KasnerSolutions}. In their work, Condeescu et al search for strictly monotone scalar field solutions corresponding to a four-parameter class of exponential potentials that are asymptotically Kasner. It should be emphasised that we only focus on whether not a solution is asymptotically Kasner. Analysing other properties of the solutions, such as stability, is beyond the scope of this work here.  
	
	The work we present here differs from \cite{KasnerSolutions} in four key ways: (1) We do not consider coupled scalar fields instead choosing to focus on the behaviour of a single \emph{scalar field}; (2) we do not a priori assume that the solutions (of the EFEs) are asymptotically Kasner. Instead we give a list of sufficient conditions for an \emph{arbitrary} potential, and prove that these conditions imply that solutions (of the EFEs) are asymptotically Kasner; (3) For a particular choice of the potential, we provide numerical examples; and (4) we do not only focus on solutions for which the scalar field is a strictly monotonic function. 
	
	The analytical results we present here focus on spatially homogenous solutions with a strictly monotonic scalar field solution. The results that we present are then numerically extended to space-times with a scalar field that is \emph{not} strictly monotonic. The assumption that the scalar field is a strictly monotonic function allows us to treat the potential as a given function of \emph{time}. We are not the first to treat the potential in this way; see, for example, \cite{GenSolArbPotential}.
	
	This paper is outlined as follows: We begin in \Sectionref{Preliminary_material} by first discussing the ADM equations (in CMC gauge and zero shift), as well as introducing the Kasner scalar field solutions. In \Sectionref{BianchiI} we restrict our attention to Bianchi I space-times, and discuss how treating the potential as a function of time can be beneficial. In \Sectionref{Section:Integral_formulas_and_the_isotropic_case} we provide evidence that the ADM equations in CMC gauge with zero shift will necessarily lead to a coordinate singularity if the scalar field is oscillatory and the potential is non-zero. In \Sectionref{Sec:Asymptotically_Kasner_solutions_in_Bianchi_I_cosmologies} we give necessary conditions that a scalar field potential must satisfy in order for the resulting Bianchi I solution to be asymptotically Kasner. In \Sectionref{Some_exact_solutions} we then give explicit examples of space-times that are and are not asymptotically Kasner. Finally, in \Sectionref{Numerical_examples} we provide numerical examples of asymptotically Kasner space-times, and in \Sectionref{Conclusions} we summarize our results.
	\section{Preliminary material}
	\label{Preliminary_material}
	\subsection[The EFEs in CMC gauge with zero shift]{The Einstein-scalar field equations in CMC gauge with zero shift}
	\label{The_Einstein-scalar_field_equations_in_CMC_gauge_with_zero_shift}
	We consider a globally hyperbolic, time-oriented oriented 4-dimensional smooth Lorentzian manifold $(M, g_{\alpha\beta})$ where $g_{\alpha\beta}$ is a smooth Lorentzian metric. Here, we study solutions of the EFEs in geometric units ($c=8\pi G=1$ for the speed of light $c$ and the gravitational constant $G$), 
	\begin{align}
	\prescript{(4)}{}{R}_{\mu\nu}-\frac{1}{2}\prescript{(4)}{}{R}g_{\mu\nu}=T_{\mu\nu},
	\label{Eq:EFEs}
	\end{align}
	where $\prescript{(4)}{}{R}_{\mu\nu},\prescript{(4)}{}{R}$ are the Ricci tensor and scalar (associated with $g_{\mu\nu}$), respectively, and $T_{\mu\nu}$ is the energy momentum tensor of the matter field. Here we consider a minimally coupled scalar field as our matter field. 
	\begin{align}
	T_{\mu\nu}=D_{\mu}\phi D_{\nu}\phi - \left( \frac{1}{2}D^{\sigma}\phi D_{\sigma}\phi + V(\phi) \right)g_{\mu\nu},
	\end{align} 
	where $D_\mu$ is the unique Levi-Civita connection associated with $g_{\mu\nu}$. The remaining freedom is the \emph{scalar field potential} $V(\phi)$. The equation of motion for the real-valued scalar field $\phi$ is 
	\begin{align}
	D^{\mu}D_{\mu}\phi=V^{\prime}(\phi),
	\label{Eq:ScalarFieldEq_4D}
	\end{align}  
	which generically follows from the divergence-free condition $D^\mu T_{\mu\nu}=0$.
	
	We now suppose that there exists a smooth function $t:M\rightarrow \mathbb{R}$ whose collection of level sets $\Sigma_t$ forms a foliation $\Sigma$ of $M$. This foliation yields a decomposition of $(M,g_{\alpha\beta})$ in the standard way. The unit co-normal of any $3$-surface $\Sigma_t\in\Sigma$ is 
	\begin{align}
	n_{\mu}=\alpha D_{\mu}t,
	\label{eq:def_Na}
	\end{align}
	where $\alpha>0$ is the \emph{lapse}. The induced first and second fundamental forms are therefore, respectively,
	\begin{align}
	\gamma_{\mu\nu}=g_{\mu\nu}+n_{\mu}n_{\nu},
	\label{Eq:Metriic_Decomp}
	\\
	K_{\mu\nu}=-\frac{1}{2}\mathcal{L}_{n}\gamma_{\mu\nu}.
	\end{align}
	The covariant derivative associated with $\gamma_{\mu\nu}$ is $\nabla_{\alpha}$. The tensor field
	\begin{align*}
	{\gamma^\mu}_{\nu}={\delta^\mu}_{\nu}+n^{\mu}n_{\nu},
	\end{align*}
	is the map that projects any tensor defined at any point  in $M$ orthogonally to a tensor that is tangent to some $\Sigma_t$. If each index of a tensor field defined on $M$ contracts to zero with $n_{\mu}$ or $n^{\mu}$, then we call that field \textit{spatial}. Given an arbitrary tensor field on $M$ we can create a spatial tensor field on $\Sigma_t$ by contracting each index with ${\gamma^\alpha}_\beta$. In fact, any tensor can be uniquely decomposed into its intrinsic and its orthogonal parts, e.g.
	\begin{align}
	\label{eq:Tdec}
	T_{\mu\nu}=\rho n_{\mu}n_{\nu}+n_{\mu}j_{\nu}+n_{\nu}j_{\mu}+S_{\mu\nu},
	\end{align}
	with
	\begin{align}
	\rho=n^{\nu}n^{\mu}T_{\mu\nu}=\frac{1}{2}\nu^2 + \frac{1}{2}\nabla_\sigma \phi \nabla^\sigma \phi + V(\phi),\quad j_{\mu}=-{\gamma^\sigma}_{\mu}n^{\nu}T_{\sigma\nu}=-\nu \nabla_\mu\phi,
	\end{align}
	and 
	\begin{align}
	S_{\mu\nu}={\gamma^\sigma}_{\mu}{\gamma^\iota}_{\nu}T_{\sigma\iota}=\nabla_{\mu}\phi\nabla_{\nu}\phi-\left( \frac{1}{2}\left(\nabla_\sigma \phi\nabla^\sigma\phi - \nu^2 \right) + V(\phi) \right)\gamma_{\mu\nu},
	\end{align}
	where we have defined 
	\begin{align}
	\nu=n^{\mu}D_{\mu}\phi.
	\end{align}
	The field $K_{\mu\nu}$ is symmetric and can be decomposed into its trace and trace-free parts (with respect to $\gamma_{\mu\nu}$) as follows
	\begin{align}
	K_{\mu\nu}=\chi_{\mu\nu}+\frac{1}{2}K \gamma_{\mu\nu}, \;\; \chi_{\mu\nu}\gamma^{\mu\nu}=0,
	\label{Eq:Decompose_Kdd_tr}
	\end{align}          
	where the relations 
	\begin{align}
	K=\gamma^{\mu\nu}K_{\mu\nu},\quad \chi_{\mu\nu}\gamma^{\mu\nu}=0,
	\end{align}
	hold and $\chi_{\mu\nu}$ is symmetric and $K$ is the \emph{mean curvature}. 
	
	Now pick an arbitrary vector field  $t^\mu$ such that 
	\begin{align}
	t^{\mu}D_{\mu}t =1.
	\end{align}
	According to \Eqref{eq:def_Na} there must exist a unique spatial vector field $\beta^{\mu}$, called the \textit{shift}, such that
	\begin{align}
	t^\mu = \alpha n^{\mu} + \beta^{\mu},
	\label{Eq:n_t}
	\end{align}
	where the quantities $\alpha$ and $\beta^\mu$ are gauge freedoms that correspond to the choice of coordinate system. For the remainder of this work we restrict our attention to the \emph{constant mean curvature} (CMC) with zero shift gauge. In particular, we set 
	\begin{align}
	K=-1/t,\quad \beta^{\mu}=0.
	\label{Eq:CMCGauge}
	\end{align}
	Notice that, as a consequence of \Eqref{Eq:CMCGauge}, the mean curvature $K$ is constant on each surface $\Sigma_t$.

	Given all this, one can decompose \Eqref{Eq:ScalarFieldEq_4D} into the following system of evolution equations 
	\begin{align}
	\partial_{t}\nu&=\alpha \Delta_{\gamma}\phi-\frac{1}{t}\alpha\nu + \gamma^{\mu\nu}\nabla_{\mu}\phi\nabla_{\nu}\alpha -\alpha V^{\prime}(\phi),
	\label{Eq:Evol_nu}
	\\
	\partial_{t}\phi&=\alpha \nu,
	\label{Eq:Evol_phi}
	\end{align}
	where $\Delta_\gamma=\gamma^{\mu\nu}\nabla_{\mu}\nabla_{\nu}$ is the Laplace-Beltrami operator associated with $\gamma_{\mu\nu}$. It is useful to note that the conjugate momentum of $\phi$ is $\pi_{\phi}=t\nu$.
	
	Similarly, from \Eqref{Eq:EFEs}, one obtains the following evolution equations for\footnote{The quantity ${\chi^\mu}_{\nu}$ is also known as the Weingarten map. We refer the interested reader to \cite{Ringstrom:SilentWaveEqs,Ringstrom:SilentGeomtry} and the references therein for more details.} ${\chi^\mu}_{\nu}$ and $\gamma_{\mu\nu}$
	\begin{align}
	\partial_{t}{\chi^\mu}_{\nu}=&-\nabla^{\mu}\nabla^{\nu}\alpha + \left( {R^{\mu}}_\nu + \left( \frac{1}{3t^2}-V(\phi)\right){\gamma^\mu}_\nu - \frac{1}{t}{\chi^\mu}_{\nu} - \nabla^\mu \phi \nabla_{\nu}\phi \right)\alpha
	\nonumber
	\\
	&-\frac{1}{3t^2}{\gamma^\mu}_{\nu},
	\label{Eq:Evol_A}
	\\
	\partial_{t}\gamma_{\mu\nu}=&-2\alpha\left( \chi_{\mu\nu}-\frac{1}{3t}\gamma_{\mu\nu} \right),
	\label{Eq:Evol_y}
	\end{align}
	where ${R^{\mu}}_\nu$ is the Ricci tensor associated with $\gamma_{\mu\nu}$. These are the \emph{ADM evolution equations}. \Eqref{Eq:EFEs} also gives rise to the \emph{constraint equations} 
	\begin{align}
	H:={R}-{\chi^\mu}_{\sigma}{\chi^\sigma}_{\mu}+\frac{2}{3}K^{2}-\nu^2 - \gamma^{\mu\nu}\nabla_{\mu}\phi\nabla_{\nu}\phi - 2 V(\phi)=0,
	\label{Eq:Constraint_H}
	\\
	M_{\mu}:=\nabla_{\sigma}{\chi^{\sigma}}_{\mu}+\nu\nabla_{\mu}\phi=0,
	\label{Eq:Constraint_M}
	\end{align}
	where $R={R^\mu}_\mu$ is the Ricci scalar associated with $\gamma_{\mu\nu}$. The fields $H$ and $M_\mu$ are the \emph{constraint violations}.
	Finally, fixing $K$ as in \Eqref{Eq:CMCGauge} gives rise to an elliptic equation for the lapse 
	\begin{align}
	\Delta_{\gamma}\alpha = \left( {R} + \frac{1}{t^2} - 3V(\phi) - \gamma^{\mu\nu}\nabla_{\mu}\phi\nabla_{\nu}\phi \right)\alpha -\frac{1}{t^2}.
	\label{Eq:LapseEq}
	\end{align}
	\newcommand{\CMCADMEqs}{\Eqsref{Eq:Evol_nu}--\eqref{Eq:LapseEq}}
	In the following we use abstract indices  $a,b,\ldots$ for $t$-dependent tensor fields on $\Sigma_t$. All indices $\mu,\nu,\ldots$ in the equations above could therefore be replaced by $a,b,\ldots$ (and at the same time each Lie-derivative along $t^\mu$ by the derivative with respect to parameter $t$).
	
	According to \cite{WellPosed:ADM}, it can be shown that given arbitrary smooth \emph{initial data} for $\gamma_{ab}$, ${\chi^{a}}_{b},\phi$ and $\nu$ (which are solutions of the constraints \Eqsref{Eq:Constraint_H} and \eqref{Eq:Constraint_M}) on an arbitrary $t=t_0$-leaf of the $(3+1)$-decomposition of $M$ the \emph{Cauchy problem} of \CMCADMEqs{} in both the \emph{increasing} and \emph{decreasing} $t$-directions is well-posed. We therefore have that \CMCADMEqs{} forms a non-linear elliptic-hyperbolic system. Note that this statement is unique to our gauge choice. In other gauges the ADM equations are generically only \emph{weakly hyperbolic} \cite{Alcubierre:Book}. 
	
	Let us now make a brief observation about \Eqref{Eq:LapseEq}. The Weingarten map ${\chi^a}_{b}$ is a trace-free tensor and as such we can, without loss of generality, set ${\chi^3}_3=-{\chi^1}_1-{\chi^2}_2$. Given this, it follows then that \Eqref{Eq:Evol_A} gives rise to \emph{two} evolution equations for ${\chi^1}_1$ (the first from the one-one component of \Eqref{Eq:Evol_A} and the second from the three-three component, with  ${\chi^3}_3=-{\chi^1}_1-{\chi^2}_2$). It is straightforward to show that these two equations are equivalent if and only if \Eqref{Eq:LapseEq} is satisfied. \Eqref{Eq:LapseEq} can therefore be thought of as an equation that `preserves' the trace-free property of the Weingarten map ${\chi^a}_{b}$.
	
	It should be emphasised here that for the remainder of this work, whenever we mention the Einstein scalar field equations we are exclusive referring to the specific formulation given by \CMCADMEqs{}.
	
	\subsection{Kasner solutions with a scalar field}
	\label{Sec:Kasner_solutions_with_a_scalar_field}

	The Kasner space-times, which can be generalised to include a scalar field \cite{RodnianskiSpeck:Linear,RodnianskiSpeck:NonLinear,KasnerSolutions}, are an example of spatially homogeneous solutions of \Eqsref{Eq:Evol_A}--\eqref{Eq:LapseEq}. These solutions play an important role in the present work and so it is prudent for us to briefly summarise their basic properties. We begin with the metric
	\begin{align}
	\gamma_{ab}=\text{diag}\left( t^{2p_1}, t^{2p_2}, t^{2p_3} \right),
	\end{align}
	which is expressed in terms of the standard Cartesian coordinates on $\Sigma_t$, where the Kasner exponents $p_1,p_2,p_3\in\mathbb{R}$ are constants which must satisfy the equations
	\begin{align}
	p_{1}+p_{2}+p_{3}=p_{1}^{2}+p_{2}^{2}+p_{3}^{2}+A^{2}=1,
	\label{Eq:KasnerRelations}
	\end{align}
	where $A$ is the \emph{``scalar field strength"} subject to the restriction $A\in[-\sqrt{2/3},\sqrt{2/3}]$. The scalar field solution is 
	\begin{align}
	\phi = A \ln\left(t\right)+B,
	\end{align}
	where $B$ is an integration constant that does not affect the dynamics of the scalar field. The lapse is 
	\begin{align}
	\alpha=1,
	\end{align}
	and the trace-free part of the extrinsic curvature is 
	\begin{align}
	{\chi^a}_{b}=-\frac{1}{t}\text{diag}\left( q_1,q_2,q_3 \right),\quad q_{i}=p_i - \frac{1}{3}.
	\end{align}
	Notice that this implies the formula
	\begin{align}
	{\chi^a}_{b}{\chi^b}_{a}=\frac{1}{t^2}\left(\frac{2}{3}-A^2\right).
	\end{align}
	Observe carefully that the requirement $A\in[-\sqrt{2/3},\sqrt{2/3}]$ implies that ${\chi^a}_{b}{\chi^b}_{a}\ge 0$. We further note that the $p_{i}$'s can be found as the eigenvalues of $t{K^{a}}_b$.
	
	We also find that the constants  $q_{i}$, $i=1,2,3$ must satisfy the constraints 
	\begin{align}
	q_{1}+q_{2}+q_{3}=0,\quad q_{1}^2+q_{2}^2 + q_{3}^2 = \frac{2}{3}-A^{2}. 
	\end{align}
	The special Kasner solution for which the space-time is isotropic, is described by setting
	\begin{align}
	p_{1}=p_{2}=p_{3}=\frac{1}{3},\quad A^{2}=\frac{2}{3},
	\end{align} 
	and corresponds to a FLRW space-time.
	\begin{figure}[t!]
		\centering
		\includegraphics[width=0.5\linewidth]{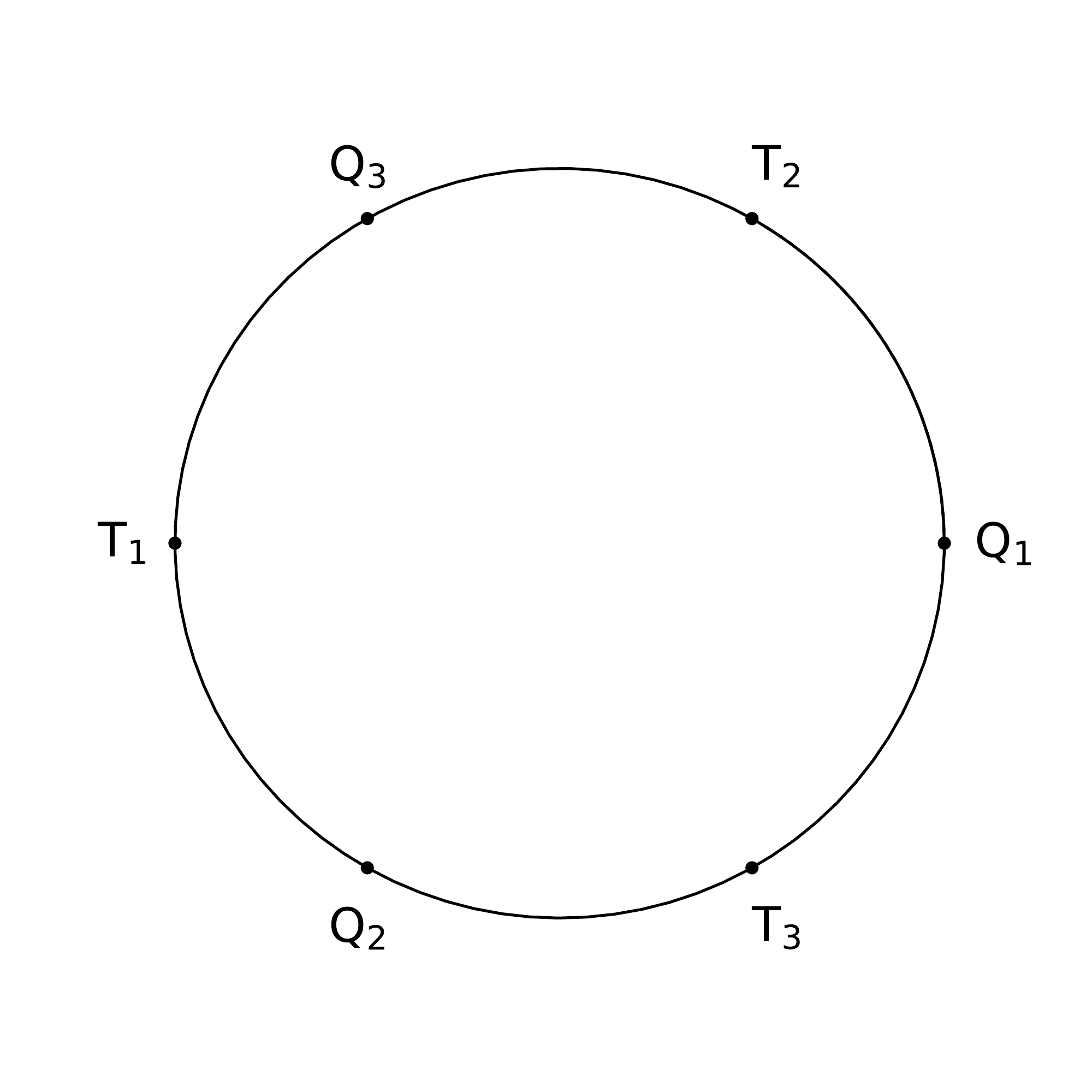}
		\caption{The Kasner circle, which shows the most vacuum solutions with $\theta=A=0$. The points labelled $T_i$ are defined as the points where $p_i=1$ for $i=1,2,3$. The midpoints between two consecutive $T_i$ and $T_k$ (for some cycle of $(ijk)=(123)$) points are labelled as $Q_j$.}
		\label{fig:kasnercircle}
	\end{figure}
	
	For later discussion it is useful to note that the set of all Kasner scalar field solutions can be expressed in terms of a three-dimensional space with coordinates $(\Sigma_+,\Sigma_-,\Sigma_0)$, which relate to the Kasner exponents as
	\begin{align}
	p_{1}=\frac{1}{3}\left( 1 - 2\Sigma_+ \right),\;\; p_{2}=\frac{1}{3}\left( 1+\Sigma_+ +\sqrt{3}\Sigma_- \right),\;\; p_{3}=\frac{1}{3}\left( 1+\Sigma_+ -\sqrt{3}\Sigma_- \right),\;\; A = \sqrt{\frac{2}{3}}\Sigma_0.
	\label{KasnerCoordinates}
	\end{align}
	In this representation the two Kasner relations reduce to a single equation:
	\begin{align}
	\Sigma_{+}^{2}+\Sigma_{-}^{2}+\Sigma_{0}^{2}=1.
	\end{align}
	It follows then that the set of Kasner scalar field solutions are represented by a unit-sphere in $\mathbb{R}^3$ and so it is useful to parametrise the coordinates as 
	\begin{align}
	\Sigma_{+}=\cos\left( \psi \right)\cos\left( \theta \right),\;\; \Sigma_{-}=\sin\left( \psi \right)\cos\left( \theta \right),\;\; \Sigma_{0}=\sin\left( \theta \right),
	\end{align}
	for angular coordinates $\psi\in\left[0,2\pi \right)$ and $\theta\in\left(-\pi/2,\pi/2 \right]$. We refer to this sphere as \emph{`the Kasner sphere'}. The Kasner sphere is the generalisation of the more famous Kasner-circle, which only shows the vacuum ($\theta=0$) solutions, to include the scalar field solutions. The Kasner circle is shown in \Figref{fig:kasnercircle}, with its most relevant points: The \emph{Taub points} $T_i$ and the \emph{locally rotationally symmetric points} $Q_i$. At a Taub point (represented as $T_{i}$ in \Figref{fig:kasnercircle}) one of the Kasner exponents is one, while the remaining two are identically zero. Locally rotationally symmetric points (represented as $Q_{i}$ in \Figref{fig:kasnercircle}) are the midpoints between two consecutive Taub points. At such a point none of the Kasner exponents are zero and two of them are equal. These points naturally divide the Kasner circle into six arcs. For any solution in one of these arcs there exists a unique isometry to one of the other arcs. It is therefore common to restrict ones attention to one of the arcs.   
	
	Notice that the exponents $p_{i}$ take their largest and smallest values when $\theta=0$. In particular we find that $p_{i}\in [ -1/3,1 ]$, and hence $q_{i}\in[ -2/3, 2/3 ]$. We shall always discuss Kasner scalar field solutions in terms of the angles $(\theta,\psi)$.
	
	\section{Bianchi I cosmologies with a strictly monotonic scalar field}
	\label{BianchiI}
	\subsection{The potential as a function of time in CMC gauge with zero shift}
	\label{Integral_formulas_for_Bianchi_I_space-times}
	The Bianchi cosmologies are a class of cosmological models that are spatially homogeneous but not necessarily isotropic. These space-times are characterised by the presence of three (spatial) Killing vectors. For the Bianchi I cosmologies, which are of particular interest in this section, the pairwise commutator of the Killing vectors is identically zero. We refer the interested reader to \cite{Wainwright:2005wss} for more in-depth discussion about Bianchi cosmologies.
	
	We say that the fields $(\gamma_{ab},{\chi^a}_b,\alpha,\phi,\nu,V(\phi))$ describe a Bianchi I cosmology if they are spatially homogeneous solutions of the Einstein scalar field equations \CMCADMEqs{} on some interval $I=\left( 0, T \right]$ for some initial time $T>0$. Assuming that the solutions are spatially homogeneous leads to a decoupling of the evolution equation for the fundamental forms $(\gamma_{ab},{\chi^a}_b)$ and the matter fields $(\nu,\phi)$. In particular we find that \Eqsref{Eq:Evol_A} and \eqref{Eq:Evol_y} become 
	\begin{align}
	\partial_{t}{\chi^a}_{b}&= - \frac{\alpha}{t}{\chi^a}_{b},
	\label{Eq:ADM_Homogenous_A}
	\\
	\partial_{t}\gamma_{ab}&=-2\alpha\left( \chi_{ab}-\frac{1}{3t}\gamma_{ab} \right).
	\label{Eq:ADM_Homogenous_y}
	\end{align}
	Similarly, \Eqsref{Eq:Evol_nu} and \eqref{Eq:Evol_phi} become 
	\begin{align}
	\partial_{t}\nu&=-\frac{1}{t}\alpha\nu -\alpha V^{\prime}(\phi),
	\label{Eq:ADM_Homogenous_nu}
	\\
	\partial_{t}\phi&=\alpha \nu.
	\label{Eq:ADM_Homogenous_phi}
	\end{align}
	In \Eqsref{Eq:ADM_Homogenous_A}--\eqref{Eq:ADM_Homogenous_phi} the lapse $\alpha$ (which is found as a solution of \Eqref{Eq:LapseEq}) is given \emph{algebra\-ically} as
	\begin{align}
	\alpha = \frac{1}{1-3t^2 V(\phi)}.
	\label{Eq:Lapse_Homogenous}
	\end{align}
	
	We point out that even though evolution equations \Eqsref{Eq:ADM_Homogenous_A}--\eqref{Eq:ADM_Homogenous_phi} have decoupled, the fields $(\gamma_{ab},{\chi^a}_b,\phi,\nu)$ are still intrinsically related via the Hamiltonian constraint \Eqref{Eq:Constraint_H}
	\begin{align}
	{\chi^a}_{b}{\chi^b}_{a} =  \frac{2}{3t^2} - \nu^2 - 2 V(\phi).
	\label{Eq:ADM_Homogenous_Hamiltonian}
	\end{align}
	
	One sees immediately from \Eqref{Eq:Lapse_Homogenous} that if there exists a $t_\star\in I$ such that $V(\phi(t_\star))=1/(3t_{\star}^2)=:V_\star$ then the lapse $\alpha$ is not defined. Moreover, we have that the lapse $\alpha$ is positive if and only if $V(\phi(t))<V_\star$ for all $t\in I$. The point $t=t_\star$ likely corresponds to a breakdown of our gauge choice (and is therefore a coordinate singularity). However, it could also correspond to a physical singularity. The only conclusive way to demonstrate that $t=t_\star$ is a coordinate singularity is to find a coordinate system such that, in the these new coordinates, the fields $(\gamma_{ab},{\chi^a}_b,\alpha,\phi,\nu,V(\phi))$ are well-defined at $t=t_\star$. Finding such a coordinate system can be difficult. In such an instance one may instead consider \emph{curvature invariants}. If a curvature invariant remains finite as $t\rightarrow t_\star$ then it implies (but does not prove) that $t=t_{\star}$ is a coordinate singularity. The most commonly considered curvature invariant is the \emph{Kretschmann scalar}. However, in the presence of matter one may instead consider the Ricci scalar or the Ricci tensor contracted with itself \cite{Ringstrom:SCC_T3Gowdy}. If any of the curvature invariants are not finite as $t\rightarrow t_\star$ then $t=t_\star$ is a physical singularity.   
	
	Suppose now that \emph{there is no} $t_\star \in I$ such that $V(\phi(t_\star))=V_\star$ (which we shall assume from now on, unless stated otherwise). Then, to solve \Eqsref{Eq:ADM_Homogenous_A} -- \eqref{Eq:Lapse_Homogenous} one typically specifies potential as a function of the \emph{scalar field} $\phi$, first. One therefore, as a matter of principle, does not a priori know how the function $V(t)=V(\phi(t))$ depends on \emph{time}. It is therefore not always possible to ensure that $V(\phi)\ne V_\star$ for all $t\in I$ \emph{before} the equations are solved. Although it should be noted that it \emph{is} possible for some potentials (for example $V(\phi)=-\phi^{2}<0$). If one instead knew the function $V(t)$, and not $V(\phi)$, then it would be possible to ensure that $V(t)\ne V_\star$ for all $t\in I$ \emph{before} \Eqsref{Eq:ADM_Homogenous_A} -- \eqref{Eq:Lapse_Homogenous} are solved. In such a setting \Eqsref{Eq:ADM_Homogenous_nu} and \Eqref{Eq:ADM_Homogenous_phi} ensure that the energy-momentum tensor is divergence free only if\footnote{If there is a $t_\star\in I$ such that $(\partial_{t}\phi)(t_{\star})=0$ \emph{and} $(\partial_{t}V)(t_{\star})=0$ then $D^\alpha T_{\alpha\beta}=0$. However, the implicit function theorem suggests that it may not be possible to calculate $V(\phi)$ at such a point.} $\partial_{t}\phi \ne 0$ for all $t\in I$. Of course, there is no way to enforce this condition when solving \Eqsref{Eq:ADM_Homogenous_A} -- \eqref{Eq:Lapse_Homogenous}. Nevertheless, if $\partial_{t}\phi\ne0$ for all $t\in I$ then the implicit function theorem ensures that it is possible to calculate $V(\phi)$, once $\phi(t)$ has been determined. If $\phi$ is such that $\partial_{t}\phi\ne0$ and $\nu\ne0$ for all $t\in I$ then we say that $\phi$ is \emph{strictly monotonic} on $I$. It is worth pointing out that $\partial_{t}\phi\ne 0$ is equivalent to $\nu\neq 0$ provided the lapse $\alpha$ is finite (and non-zero). Often we refer to the solutions $(\gamma_{ab},{\chi^a}_b,\alpha,\phi,\nu,V(\phi))$, with a strictly monotonic scalar field $\phi$ as a strictly monotonic solution. Treating the potential as a function of time is something of a trade off, as we can now a priori ensure that $V(t)\ne V_\star$ for all $t\in I$, but we cannot study solutions that do not have a strictly monotonic scalar field $\phi$. Nevertheless, this approach shall play a fundamental role in our analytical treatment of the potential.
	
	We emphasize that this approach is formally consistent only in the spatially homogeneous case. In particular we point out that it is, in general, not possible to pick the potential $V(\phi)$ as a function of the coordinates in the spatially inhomogeneous setting. Generically, if the potential is chosen as a function of the coordinates then \Eqref{Eq:ScalarFieldEq_4D} \emph{does not} follow from the divergence-free condition $D^\mu T_{\mu\nu}=0$. Even in the spatially homogeneous setting this approach makes sense only if the scalar field $\phi$ is \emph{assumed} to be strictly monotonic.

	We now discuss how the evolution equation for $\nu$ (\Eqref{Eq:ADM_Homogenous_nu}) changes when the potential $V(t)$ is a known function of time. To this end, we suppose that $\phi$ is strictly monotonic on $I$. Then, we can write
	\begin{align}
	\frac{\partial V}{\partial t}=\frac{\partial V}{\partial \phi}\frac{\partial \phi}{\partial t}=\alpha\nu \frac{\partial V}{\partial \phi}\implies\alpha V^{\prime}(\phi)=\frac{1}{\nu}\partial_{t}V.
	\label{Eq:V_of_t}
	\end{align}
	Notice that if $\phi$ is not strictly monotonic on $I$ (and hence there is a $t_{\dagger}\in I$ such that $(\partial_{t}\phi)(t_{\dagger}) = \nu(t_{\dagger})=0$) then it would not be possible to manipulate the potential in this way. \Eqref{Eq:V_of_t} now allows us to re-express the scalar field equation \Eqref{Eq:Evol_nu} as  
	\begin{align}
	t\partial_{t}\left( {\nu^{2}} \right) =-2\left( \alpha\nu^{2} + t\partial_{t}V(t) \right).
	\label{Eq:NuEq_V(t)}
	\end{align}
	Note that one can guarantee $\partial_{t}\phi \neq 0$ only \emph{after} \Eqref{Eq:NuEq_V(t)} is solved.

	\subsection[Integral formulas and the isotropic case]{Bianchi I integral formulas for a strictly monotonic scalar field and the isotropic case}
	\label{Section:Integral_formulas_and_the_isotropic_case}
	In this section here we consider generic solutions of \Eqsref{Eq:ADM_Homogenous_A}--\eqref{Eq:Lapse_Homogenous} in terms of an arbitrary potential $V(t)$ (such that $V(t)<V_{\star}$ for all $t\in I$) with a strictly monotonic scalar field $\phi$. 
	\begin{Thm}
		\label{Thm_1}
		Consider the spatially homogeneous fields $(\gamma_{ab},{\chi^a}_b,\phi,\nu,\alpha,V(t))$ which are solutions of \Eqsref{Eq:ADM_Homogenous_A}--\eqref{Eq:Lapse_Homogenous}, and suppose that the scalar field $\phi$ is monotonic on the interval $I=\left( 0, T \right]$ for some $T>0$ and $V(t)<1/(3t^2)$ for all $t\in I$. Then, under these restrictions,
		the general solution of \Eqref{Eq:ADM_Homogenous_A} (\Eqref{Eq:Evol_A}) is
		\begin{align}
		{\chi^a}_b =  {C^a}_b \exp\left( f(t) \right),\quad f(t)=-\int_{T}^t \frac{1}{\tau\left( 1-3\tau^2 V(\tau) \right)} d\tau,
		\label{Eq:GeneralSolution_A}
		\end{align}	
		and the general solutions of \Eqsref{Eq:ADM_Homogenous_nu} and \eqref{Eq:ADM_Homogenous_phi} (\Eqsref{Eq:Evol_nu} and \eqref{Eq:Evol_phi}) are 
		\begin{align}
		\nu^{2}=\frac{2}{3t^2} \left( 1 - 3 t^2 V(t) \right) + \C_\nu \text{e}^{2 f\left( t \right)},
		\quad 
		\phi=\C_\phi + \int_{T}^t \frac{\nu(\tau)}{1-3\tau^2 V(\tau)} d\tau,
		\label{Eq:General_nu}
		\end{align}
		where $\C_\nu,\C_\phi$ are integration constants and ${C^a}_b$ is a trace-free (${C^a}_a =0$), symmetric ($C_{ab}=C_{ba}$) tensor with constant entries. Moreover, the fields $({C^a}_b,\C_\nu)$ satisfy the ``Kasner constraint''
		\begin{align}
		{C^a}_b {C^b}_a + \C_\nu =0.
		\label{Eq:GeneralKasnerConstraint}
		\end{align}
		Furthermore, if $\C_\nu =0$ then ${C^a}_{b}=0$ and 
		\begin{align}
		\gamma_{ab}=\tilde{\gamma}_{ab}\exp\left( -\frac{2}{3}f(t) \right),
		\label{Eq:IsotropicMetric_WithPotential}
		\end{align}
		where $\tilde{\gamma}_{ab}$ is a symmetric tensor ($\tilde{\gamma}_{ab}=\tilde{\gamma}_{ba}$) with constant entries.
	\end{Thm}
	
	\begin{proof}
		It is straightforward to show that \Eqref{Eq:GeneralSolution_A} is the general solution of \Eqref{Eq:ADM_Homogenous_A}; one only needs to note that 
		\begin{align}
		\frac{df(t)}{dt}=-\frac{1}{t\left( 1- 3 t^2 V(t) \right)}=-\frac{\alpha}{t}.
		\end{align}
		Similarly, by direct calculation one easily checks that \Eqsref{Eq:General_nu} solves \Eqsref{Eq:ADM_Homogenous_phi}.
		
		Inputting the solutions \Eqref{Eq:GeneralSolution_A} and \Eqsref{Eq:General_nu} into the Hamiltonian constraint \Eqref{Eq:ADM_Homogenous_Hamiltonian}, gives 
		\begin{align}
		\left( {C^a}_b {C^b}_a + \C_\nu \right)\text{e}^{2f(t)}=0,
		\end{align}
		and hence \Eqref{Eq:GeneralKasnerConstraint} must hold. Suppose now that $\C_\nu=0$. Then, from \Eqref{Eq:GeneralKasnerConstraint} we find that we must have ${C^a}_b {C^b}_a=0$ which holds if and only if ${C^a}_b={\chi^a}_b=0$. Then, \Eqref{Eq:Evol_y} becomes
		\begin{align}
		\partial_t\gamma_{ab}=\frac{2}{3(1-3t^2 V(t))t}\gamma_{ab}.
		\label{Eq:IsotropicMetricEq}
		\end{align} 
		It is now straightforward to show that \Eqref{Eq:IsotropicMetric_WithPotential} is the general solution of \Eqref{Eq:IsotropicMetricEq}.
	\end{proof}
	There are three interesting notes to be made about \Thmref{Thm_1}: Firstly, the function $f(t)$ can be given geometrical significance by noting that $f(t)=\ln(\gamma(t))/2$ where $\gamma=\det(\gamma_{ab})$. Secondly, the \emph{``Kasner constraint''} \Eqref{Eq:GeneralKasnerConstraint} earns its name since it reduces to the standard Kasner constraint, presented in \Sectionref{Sec:Kasner_solutions_with_a_scalar_field}, when the potential $V(t)$ is identically zero. However, if $V(t)\ne 0$, then the fields $({C^a}_b,\C_\nu)$ only retain their geometrical significance if the potential $V(t)$ satisfies particular decay conditions. We discuss this further in \Sectionref{Sec:Asymptotically_Kasner_solutions_in_Bianchi_I_cosmologies}. Thirdly, in the special case $\C_\nu =0$ the exact solution for $\nu$ is $\nu=\nu_\star$, where $\nu_{\star}$ is defined such that 
	\begin{align}
	\nu_{\star}^2= \frac{2}{3t^2} \left( 1 - 3 t^2 V(t) \right)=\frac{2}{3t^2 \alpha},
	\label{SpecialSolutionForPi}
	\end{align} 
	which follows directly from \Eqref{Eq:General_nu}. Moreover, \Thmref{Thm_1} tells that, in this setting, the metric is isotropic. We therefore interpret this solution (corresponding to $\C_\nu=0$) as the generalisation of the FLRW solution (with zero potential) discussed in \Sectionref{Sec:Kasner_solutions_with_a_scalar_field} to include a (possibly non-zero) potential $V(t)$. We further support this interpretation by noting that if $\C_\nu=0$ and $V(t)=0$ then the isotropic Kasner scalar field solution (discussed in \Sectionref{Sec:Kasner_solutions_with_a_scalar_field}) is returned. In fact, if one first assumes that the metric $\gamma_{ab}$ is isotropic, then it is possible to show that \Eqref{SpecialSolutionForPi} still gives the exact solution for $\nu$ with $V(t)$ replaced with $V(\phi)$. This isotropic solution is particularly important as it is an explicit example that can be used to offer insight into the possibility of a gauge breakdown.  
	
	\begin{Thm}
		\label{Corollary}
		Consider the spatially homogeneous fields $(\gamma_{ab},{\chi^a}_b,\phi,\nu,\alpha,V(\phi))$ which are solutions of the Einstein scalar field equations \Eqsref{Eq:ADM_Homogenous_A}--\eqref{Eq:Lapse_Homogenous} on the interval $I=\left( 0, T \right]$ for some $T>0$ such that, at $t=T$, we have 
		\begin{align}
		\nu(T)^2=\frac{2}{3 T^2}\left( 1- 3 T^2 V(\phi(T)) \right),
		\label{Eq:Coro_IDChoice}
		\end{align}
		with $\nu(T)\ne 0$ and $T^2 V(\phi(T))<1/3$. Then, $V(\phi(t))<1/(3t^2)$ for all $t\in I$ if and only if $\phi(t)$ is strictly monotonic on $I$.
	\end{Thm}
	We remark here that since the fields $(\gamma_{ab},{\chi^a}_b,\phi,\nu,\alpha,V(\phi))$ in \Thmref{Corollary} are solutions of the EFEs, \Eqref{Eq:Coro_IDChoice} together with the Hamiltonian constraint \Eqref{Eq:ADM_Homogenous_Hamiltonian} implies that \[{\chi^a}_{b}(T)=0.\] 
	\begin{proof}
		Suppose first that the scalar field $\phi$ is not strictly monotonic on $I$. Then, there exists at least one time $t_0\in \left( 0, T \right)$ such that $\phi$ is strictly monotonic on the interval $I_0 := \left( t_{0}, T \right]$ with $\nu(t_0)=(\partial_{t}\phi)(t_0)=0$. Then, from \Eqref{Eq:Coro_IDChoice} it follows that $\nu(t)=\nu_{\star}(t)$ for all $t\in I_0$, where $\nu_\star(t)$ is given by \Eqref{SpecialSolutionForPi}, and hence 
		\begin{align}
		\nu^{2}= \frac{2}{3 t^2}\left( 1- 3 t^2 V(\phi(t)) \right)\quad \text{for all $t\in I_0$}.
		\label{Eq:Coro_dphi}
		\end{align}
		Recall that $\nu\rightarrow 0$ in the limit $t\rightarrow t_0^+$. Then, from \Eqref{Eq:Coro_dphi} we see that as $t\rightarrow t_0^+$ we must have $V(\phi(t))\rightarrow 1/(3t_0^2)$, as was claimed.

		Suppose now that $\phi(t)$ is a strictly monotonic function. Then, as noted previously, the exact solution is $\nu(t)=\nu_\star(t)$ for all $t\in I$, and hence $\nu$ is given by \Eqref{Eq:Coro_dphi} with $t_0=0$. Now, if there exists a point $t_\star$ such that $V(\phi(t_\star))= 1/(3t_\star^2)$ then \Eqref{Eq:Coro_dphi} tells that $\nu(t_\star)=0$ which contradicts our initial assumption that $\phi$ is a strictly monotonic function. It follows then that such a $t_\star$ cannot exist.
	\end{proof}
	At first glance \Thmref{Corollary}, appears only to be a statement about the effect of a particular choice of initial data for the function $\nu$. However, it must be emphasized that picking $\nu(T)^2 = \nu_{\star}(T)^2$ is the \emph{only} choice of initial data that gives rise to an isotropic solution. Furthermore, we note that \Thmref{Corollary} relies on the fact that $\nu(t)=\nu_{\star}(t)$ for all $t\in \left[ T, t_0 \right)$. It follows that if $\C_\nu$ (see \Eqref{Eq:General_nu}) is small then we \emph{do not} expect \Thmref{Corollary} to hold.
	
	Although this result shows that a singularity exists (for scalar fields $\phi$ that are not strictly monotonic on an isotropic space-time) it says nothing about the nature of the singularity itself. To understand whether or not this is a coordinate singularity we consider the following curvature invariants:
	\begin{align}
	\prescript{(4)}{}{R}=4V(\phi(t))-\nu^{2}&,
	\quad
	\prescript{(4)}{}{R}^{\mu\nu}\prescript{(4)}{}{R}_{\mu\nu}=\nu^{4}+4V(\phi(t))^{2}-2\nu^{2}V(\phi(t)),
	\end{align}
	and 
	\begin{align}
	\prescript{(4)}{}{R}^{\mu\nu\alpha\beta}\prescript{(4)}{}{R}_{\mu\nu\alpha\beta}=\frac{2}{27t^{4}}+\frac{4}{3t^{2}}\left( \frac{1}{\alpha} - \frac{5}{3t^2} \right)^{2}.
	\end{align}
	Taking the limit towards the singularity gives
	\begin{align}
	\prescript{(4)}{}{R}\rightarrow \frac{4}{3t^{2}_\star},
	\quad
	\prescript{(4)}{}{R}^{\mu\nu}\prescript{(4)}{}{R}_{\mu\nu}\rightarrow \frac{4}{9t^{4}_\star},
	\quad
	\prescript{(4)}{}{R}^{\mu\nu\alpha\beta}\prescript{(4)}{}{R}_{\mu\nu\alpha\beta}\rightarrow \frac{2}{27t_\star^{4}}\left( 1 +\frac{50}{t_\star^{2}}\right),
	\quad
	\text{as}
	\quad
	t\rightarrow t_\star.
	\end{align}
	This strongly suggests (but does not prove) that $t=t_\star>0$ is a coordinate singularity. Recall that the only conclusive way to show that $t=t_\star$ is a physical singularity is to find a coordinate system for which $t=t_\star$ is not a singularity. Note here that we interpret \Thmref{Corollary} as implying that CMC gauge with zero shift is poorly suited to this particular type of space-time.

	\subsection{Asymptotically Kasner solutions in Bianchi I cosmologies}
	\label{Sec:Asymptotically_Kasner_solutions_in_Bianchi_I_cosmologies}
	In this section here we now discuss what properties a time-dependent potential must possess if the corresponding (spatially homogeneous) solutions are to be \emph{asymptotically Kasner}, a notion that we define in the following way:

	\begin{defn}
		\label{Def:AsymptoticallyKasner}
		Consider the spatially homogeneous fields $(\gamma_{ab},{\chi^a}_b,\phi,\nu,\alpha,V(\phi))$ which are solutions of the Einstein scalar field equations \Eqsref{Eq:ADM_Homogenous_A}--\eqref{Eq:Lapse_Homogenous} on the interval $I=\left( 0, T \right]$ with $T>0$. Then, we say that the fields $(\gamma_{ab},{\chi^a}_b,\phi,\nu,\alpha,V(\phi))$ are ``{asymptotically Kasner}" if the limits
		\begin{align}
		\lim_{t\rightarrow 0^{+}}t{\chi^a}_b = {\C^a}_b, \quad \lim_{t\rightarrow 0^{+}}t\nu = A,\quad \lim_{t\rightarrow 0^{+}}\alpha = 1,
		\label{Eq:Def_Limits}
		\end{align}
		exist, where $A\in [-\sqrt{2/3},\sqrt{2/3}]$  and ${\C^a}_b{\C^b}_a \in \left[ 0, 2/3 \right]$ such that
		\begin{align}
		{\C^a}_b{\C^b}_a + A^2 = \frac{2}{3}.
		\end{align}
		If a spatially inhomogeneous solution $(\gamma_{ab},{\chi^a}_b,\phi,\nu,\alpha)$ of the Einstein scalar field equations \Eqsref{Eq:Evol_A}--\eqref{Eq:LapseEq} is asymptotically Kasner for each fixed spatial point $p\in\Sigma$ then we say that the solution $(\gamma_{ab},{\chi^a}_b,\phi,\nu,\alpha)$ is ``{asymptotically point-wise Kasner}".
	\end{defn}
	Notice that the definition presented here is less restrictive than the one given in \cite{LapseScalarField}. In fact, the two definitions are \emph{not} equivalent: \cite{LapseScalarField} imposes further restrictions on ``how fast" the limits \Eqref{Eq:Def_Limits} are obtained. As such, it is possible that what we claim to be asymptotically Kasner, is \emph{not} regarded as such by \cite{LapseScalarField}. We do not discuss this any further as \Defref{Def:AsymptoticallyKasner} is sufficient for our purposes here.
	
	Loosely speaking, \Defref{Def:AsymptoticallyKasner} tells us that if a solution is asymptotically (point-wise) Kasner then it can be \emph{matched} (point-wise) to a Kasner scalar field solution with zero potential (see \Sectionref{Sec:Kasner_solutions_with_a_scalar_field}). To discuss this further we note that the Hamiltonian constraint \Eqref{Eq:ADM_Homogenous_Hamiltonian} can (in the spatially homogeneous setting) be written as 
	\begin{align}
	t^2 {\chi^a}_b{\chi^b}_a + t^2 \nu^2 =\frac{2}{3} + 2 t^2 V(\phi(t)).
	\label{Eq:RelativeHamiltonian}
	\end{align}
	From \Defref{Def:AsymptoticallyKasner} we find that the quantities $t\nu$ and $t{\chi^a}_b$ are bounded functions of time, and the fields $({\C^a}_b,A)$ are the asymptotic values of $t{\chi^a}_b$ and $t\nu$, respectively, at $t=0$. Of course, this interpretation only makes sense if the potential $V(\phi)$ decays sufficiently fast. From \Eqref{Eq:RelativeHamiltonian} one immediately sees that we require $t^2 V(\phi)\rightarrow 0$ as $t\rightarrow 0^+$. This observation can be formulated more rigorously, for a strictly monotonic scalar field $\phi$, in the following way:

	\begin{Thm}
		\label{Result:Kanser}
		Let $V(t)$ be $C^{1}(I)$ on an interval $I=\left( 0, T \right]$ for some fixed constant $T>0$ and suppose that $\partial_{t}\phi\ne 0$ for all $t\in I$ and that there is no $t_{\star}\in I$ such that $V(t_\star)=1/(3t_{\star}^2)$. Further more, suppose that there exists a constant $\epsilon>0$ such that the quantity $U(t)=t^{2-\epsilon}V(t)$ is bounded on $I$. Then,
		\begin{itemize}
			\item[(1)] if the solutions to the Bianchi I scalar field equations \Eqsref{Eq:ADM_Homogenous_A} -- \eqref{Eq:Lapse_Homogenous} are asymptotically Kasner then 
			\begin{align}
			\text{if}\;\; t\rightarrow 0^{+}\;\; \text{then}\;\; t^{2}V(t)\rightarrow 0.
			\end{align}
			\item[(2)] Moreover,
			\begin{align}
			\text{if}\;\; t\rightarrow 0^{+}\;\; \text{then}\;\; U(t)\rightarrow 0,
			\label{Eq:KasnerCondition}
			\end{align}
			then the solutions to the Bianchi I scalar field equations \Eqsref{Eq:ADM_Homogenous_A} -- \eqref{Eq:Lapse_Homogenous} are asymptotically Kasner
		\end{itemize} 
	\end{Thm}
	Note that \Eqref{Eq:KasnerCondition} can be weakened so that $U(t)\rightarrow constant$ as $t\rightarrow 0$. However, one can always pick a slightly smaller $\epsilon$ so that $U(t)\rightarrow 0$ as $t\rightarrow 0$. We can therefore assume \Eqref{Eq:KasnerCondition} without loss of generality.
	\begin{proof}
		Let us first assume that the solutions are asymptotically Kasner. Then, from \Defref{Def:AsymptoticallyKasner}, we have that the limits 
		\begin{align}
		C^{a}_{b}=\lim_{t\rightarrow 0^+}t {\chi^a}{}_{b},
		\quad
		A=\lim_{t\rightarrow 0^+}t \nu,
		\end{align}
		are well defined and finite. Moreover, the (constant) fields $(A,C^{a}{}_{b})$ satisfy the algebraic relation
		\begin{align}
		{\C^a}_b{\C^b}_a + A^2 = \frac{2}{3}.
		\end{align}
		Now, taking the limit $t\rightarrow 0^+$ of the Hamiltonian constraint \Eqref{Eq:RelativeHamiltonian} we get
		\begin{align}
		{\C^a}_b{\C^b}_a + A^2 = \frac{2}{3} + \lim_{t\rightarrow 0^{+}} \left( t^{2}V(t) \right) \implies \lim_{t\rightarrow 0^{+}} \left( t^{2}V(t) \right)=0.
		\end{align}
		This proves the statement. Suppose now that $\lim_{t\rightarrow 0^+}U(t)=0$ and note that, if $\phi$ is strictly monotonic on $I$, then, from \Thmref{Thm_1} we see that a solution set is asymptotically Kasner in the sense of \Defref{Def:AsymptoticallyKasner} provided that 
		\begin{align}
		\text{if $t\rightarrow 0^+$ then $t\text{e}^{f(t)}\rightarrow 0$ and $\alpha\rightarrow 1$}, 
		\end{align}
		where $f(t)$ was defined in \Eqref{Eq:GeneralSolution_A}. To show that these limits hold we define the function $g(t)$ as
		\begin{align}
		g(t):= \frac{3 U(t) }{1-3 t^\epsilon U(t)},
		\end{align}
		so that the lapse $\alpha$ can be written as
		\begin{align}
		\alpha = 1 + t^{\epsilon}g(t).
		\end{align}
		It follows from \Eqref{Eq:KasnerCondition} that $g(t)\rightarrow 0$ as $t\rightarrow 0^+$, and in particular $\alpha\rightarrow 1$, as $t\rightarrow 0^+$, as required. Similarly, the function $f(t)$ (defined in \Eqref{Eq:GeneralSolution_A}) can be written as
		\begin{align}
		f(t)=-\int_{T}^t \frac{1}{\tau\left( 1-3\tau^2 V(\tau) \right)} d\tau=-\ln(t)-\int_{T}^{t}\tau^{\epsilon-1}g(\tau)d\tau
		\label{Eq:f_Expansion}
		\end{align}
		so that 
		\begin{align}
		t\exp({f(t)})=\exp\left( -\int_{T}^{t}\tau^{\epsilon-1}g(\tau)d\tau \right).
		\end{align}
		It therefore only remains to show that $\int_{T}^{t}\tau^{\epsilon-1}g(\tau)d\tau$ is bounded for all $t\in I$. For this, we note that, by assumption, the constants $M_1=\sup_{t\in I}g(t)$ and $M_2=\inf_{t\in I}g(t)$ are well-defined and finite on the interval $I$. Then
		\begin{align}
		\int_{T}^{t} \tau^{\epsilon -1}g(\tau)d\tau \le \int_{T}^{t} \tau^{\epsilon -1}M_{1}d\tau = \frac{M_1}{\epsilon}\left( t - T \right),
		\label{Eq:KasnerProof_sup}
		\end{align}
		and
		\begin{align}
		\int_{T}^{t} \tau^{\epsilon -1}g(\tau)d\tau \ge \int_{T}^{t} \tau^{\epsilon -1}M_{2}d\tau = \frac{M_2}{\epsilon}\left( t - T \right).
		\label{Eq:KasnerProof_inf}
		\end{align} 
		By taking the limit $t\rightarrow 0^+$ of \Eqsref{Eq:KasnerProof_sup} and \eqref{Eq:KasnerProof_inf} we find
		\begin{align}
		-\frac{M_2}{\epsilon}T\le\int_{T}^{0^+} \tau^{\epsilon -1}g(\tau)d\tau\le -\frac{M_1}{\epsilon}T,
		\label{g_Limit}
		\end{align}
		and hence we have that the constant $\C=\int_{T}^{0^+} \tau^{\epsilon -1}g(\tau)d\tau$ is finite, and in particular we have that 
		\begin{align}
		\lim_{t\rightarrow 0^+}t\, \text{e}^{f(t)}=\text{e}^{-\C}.
		\end{align}
		We therefore conclude the fields $(\gamma_{ab},{\chi^a}_b,\alpha,\phi,\nu,V(t))$ are asymptotically Kasner in the sense of \Defref{Def:AsymptoticallyKasner}. We now conclude our proof by calculating the remaining two limits.
		
		Inputting \Eqref{Eq:f_Expansion} into \Eqref{Eq:General_nu} gives
		\begin{align}
		\nu^{2}=\frac{2}{3t^2} \left( 1 - 3 t^2 V(t) \right) + \frac{1}{t^2}\text{exp}\left({ \int_T^t \tau^{\epsilon-1}g(\tau)d\tau }\right).
		\end{align}
		\Eqref{g_Limit} now tells us that if $t\rightarrow 0^+$ then
		\begin{align}
		\lim_{t\rightarrow 0^+}t^2\nu^2  =\frac{2}{3}+\text{e}^{-2\C}.
		\end{align}
		Thus, $t \nu\rightarrow A$ as $t\rightarrow 0^+$ for $A^2 =\text{e}^{-2\C} + 2/3$, as required. Similarly, by inputting \Eqref{Eq:f_Expansion} into \Eqref{Eq:GeneralSolution_A} we find
		\begin{align}
		{\chi^a}_b =  \frac{{C^a}_b}{t}\exp\left( - \int_{T}^t \tau^{\varepsilon-1}g(\tau)d\tau \right).
		\end{align}
		Using \Eqref{g_Limit} to calculate the limit $t\rightarrow0^+$ we find
		\begin{align}
		\lim_{t\rightarrow 0^+} t{\chi^a}_b = {C^a}_b \lim_{t\rightarrow 0^+}\exp\left( - \int_T^t \tau^{\varepsilon-1}g(\tau)d\tau\right)={\text{e}^{-\C}}{{C^a}_b}.
		\end{align}
		We therefore have that all of the limit conditions \Eqref{Eq:Def_Limits} are satisfied. This completes the proof.
	\end{proof}
	For the remainder of this work we refer to \Eqref{Eq:KasnerCondition} as the \textit{Kasner condition}. It is interesting to note that $V(t)$ can be a singular function of time, but only mildly so. 
	
	\Thmref{Result:Kanser} introduces a natural separation of the types of asymptotically Kasner solutions: Consider a spatially homogeneous solution $(\gamma_{ab},{\chi^a}_b,\phi,\nu,\alpha,V(t))$ defined on an interval $I = \left( 0, T \right]$ with $T>0$. Then, the scalar field is either (1) strictly monotonic for all $t\in I$; (2) eventually-monotonic; $\phi$ is not strictly monotonic on $I$ but there exists an interval $\left( 0, \tau \right]$ with $0<\tau<T$ such that $\partial_{t}\phi\ne 0$ for all $t\le\tau$ or, (3) asymptotically stationary; in which case we have $\partial_{t}\phi\rightarrow 0$ as $t\rightarrow 0^+$.

	\section{Some exact solutions}
	\label{Some_exact_solutions}
	In this section here we now provide to exact solutions of the Einstein scalar field equations \Eqsref{Eq:ADM_Homogenous_A}--\eqref{Eq:Lapse_Homogenous}. In particular, we provide two exact solutions: One that \emph{is} asymptotically Kasner, and one that \emph{is not}. Although these solutions are useful as explicit examples, it is not clear whether or not they have any physical relevance. We will not discuss this any further here.
	\subsection{An exact asymptotically Kasner solution with an unbounded pot\-ential}
	\label{SubSec:An_exact_Kasner-like_solution_with_a_singular_potential}
	It is remarkable consequence of \Thmref{Result:Kanser} that the potential $V(t)$ can be an unbounded function of time. The goal of the present subsection is to investigate solutions for which $V(t)$ becomes infinite in the limit $t\rightarrow 0^+$. In particular, we search for a strictly monotonic solution on the interval $I=\left( 0, T \right]$ with $T=1$. For this, we choose the simple function 
	\begin{align}
	V(t) = \frac{\V}{3t},
	\label{Eq:SimplePotential}
	\end{align}
	where $\mathcal{V}\in\mathbb{R}$ is a freely specifiable constant. For this particular choice of potential the Kasner condition (\Eqref{Eq:KasnerCondition}) is satisfied for any $\epsilon \in\left( 0,1 \right)$. The lapse is therefore 
	\begin{align}
	\alpha=\frac{1}{1 - \V t}.
	\label{Eq:Lapse_YaKasnerPot}
	\end{align}
	Notice that if we want to guarantee that the lapse is finite (and positive) for all $t\in I$ then we must make the restriction $\mathcal{V}< 1$. It is important to note here that, if $\V>0$, our time interval is always restricted to $t\in \left( 0,\V^{-1} \right)$. In particular we have that \Thmref{Result:Kanser} holds for all $T<\V^{-1}$ We therefore find that $t$ can be extended to $\infty$ if and only if $\V\le0$. Turning our attention to the extrinsic curvature equation \Eqref{Eq:ADM_Homogenous_A} we use \Eqref{Eq:GeneralSolution_A} to find,   
	\begin{align}
	\begin{split}
	{\chi^a}_{b}&=\exp\left( \ln \left({C^a}_b\right)+ \ln(1-\V) -\int_{1}^t \frac{1}{\tau(1-\mathcal{V}\tau)}d\tau \right) = \left( \frac{1}{t} - \mathcal{V} \right){\mathcal{C}^{a}}_{b},
	\end{split}
	\end{align}
	where ${\mathcal{C}^{a}}_{b}$ is a symmetric (i.e. $C_{ab}=C_{ba}$) tensor with constant entries subject to the requirement ${\mathcal{C}^{a}}_{a}=0$. Before proceeding we make the further simplification that the extrinsic curvature and metric are both diagonal. i.e. 
	\begin{align}
	{\chi^a}_b={\mathcal{C}^a}_b=0,\quad \text{and}\quad \gamma_{ab}=0\quad \text{if}\quad a\ne b.
	\end{align}
	We note that this is not a significant restriction as, for spatially homogeneous solutions, one can always perform a (local) coordinate transformation so that ${\chi^a}_{b}$ and $\gamma_{ab}$ are diagonal.
	
	Solving \Eqref{Eq:ADM_Homogenous_y} for the metric components now gives
	\begin{align}
	\gamma_{ab}=\frac{1}{\left( 1-\mathcal{V}t \right)^{2/3}} \text{diag}\left( \gamma_{1}t^{\frac{2}{3}-2{\mathcal{C}^1}_{1}}, \gamma_{2}t^{\frac{2}{3}-2{\mathcal{C}^2}_{2}}, \gamma_{3}t^{\frac{2}{3}-2{\mathcal{C}^3}_{3}} \right),
	\end{align}
	where $\gamma_{1},\gamma_{2}$ and $\gamma_{3}$ are integration constants. In principle these can be any constants, however this just corresponds to a scaling of the spatial coordinates. Notice that these solutions are very similar to the Kasner metric and in fact reduce to it if we pick $\mathcal{V}=0$. The resemblance becomes clearer if we set 
	\begin{align}
	{C^{a}}_{b}=\text{diag}\left(\frac{1}{3}-p_{1},\frac{1}{3}-p_{2},\frac{1}{3}-p_{3}\right) \implies \sum_{i=1}^{3}p_{i}=1,
	\end{align}
	for constants $p_{i}\in\mathbb{R}$.

	Before directly solving for the scalar field $\phi$ we must first calculate $\nu^2$. From \Eqref{Eq:General_nu}, we find
	\begin{align}
	\nu^{2}=\frac{1}{t^2 \alpha^{2}}\left(\frac{2}{3}\alpha - \left( \frac{2}{3}-A^2 \right) \right),
	\label{Eq:Nu_YaKasnerPot}
	\end{align}
	where $A\in\mathbb{R}$ is an integration constant and $\alpha$ is given by \Eqref{Eq:Lapse_YaKasnerPot}. The main thing to notice here is that, depending on the choice of $A$, it may be possible for $\nu^2$ to become negative. In order to avoid this, we must choose a region on which we want the solution to be defined. Since our primary interest is the approach to the singularity, we choose $t\in \left( 0, 1 \right]$. In the limit $t\rightarrow 0^+$ we find that the restriction $A^2\ge 0$ should be made. Furthermore, requiring the solution to be defined at $t=1$ gives 
	\begin{align}
	\begin{cases}
	1>\V > \frac{3A^2}{3A^2 - 2},\quad &A^2\ne 2/3,
	\\
	1>\V > -\infty,\quad &A^2= 2/3,
	\end{cases}
	\label{V_Limits}
	\end{align}
	The strict inequality in \Eqref{V_Limits} ensures that $\nu$ (and hence $\partial_{t}\phi$) never vanishes. We are now in a position to solve \Eqref{Eq:Evol_phi} for the scalar field $\phi$. Using \Eqref{Eq:General_nu} we find
	\begin{align}
	\phi=\begin{cases}
	B\,\ \;&\text{if}\; A=0,
	\\
	\phi_1 + \sqrt{\frac{8}{3}-4A^2} \arctan\left( \frac{\sqrt{3}\, \alpha t|\nu|}{\sqrt{2-3 A^2}} \right)+B\,\ \;&\text{if}\; A \in \left( 0, \sqrt{2/3} \right) ,
	\\
	\phi_{1}+ B\,\ \;&\text{if}\; A= \sqrt{2/3},
	\end{cases}
	\label{ScalarFieldEquation}
	\end{align}
	where $B\in \mathbb{R}$ is an integration constant and $\phi_1$ is defined as 
	\begin{align}
	\begin{split}
	\phi_{1}=&A\left(\mu\ln\left( \mu \sqrt{3}\left( t|\nu|\alpha-A\right) \right)- \ln\left( \sqrt{3}\left(t|\nu|\alpha + A\right) \right)\right),\quad \mu = \text{sign}\left( \V \right).
	\end{split}
	\end{align}
	Here $\alpha$ is given by \Eqref{Eq:Lapse_YaKasnerPot} and $\nu$ by \Eqref{Eq:Nu_YaKasnerPot}. Note that we have only considered the case $\nu>0$. The case where $\nu<0$ is recovered by the transformation $(\phi,\nu)\mapsto (-\phi,-\nu)$. 
	
	Finally, from \Thmref{Thm_1}, we find that the Hamiltonian constraint is satisfied if and only if the relation 
	\begin{align}
	\sum_{i=1}^{3}p_{i}^2=1-A^{2}
	\end{align}
	holds. Notice here that ${\chi^a}_{b}{\chi^b}_{a}\ge0 \implies  A^2\le 2/3$. Clearly the standard Kasner relations are satisfied. In particular this solution is asymptotically Kasner in the sense of \Defref{Def:AsymptoticallyKasner}. It follows then that this solution represents a generalisation of the standard Kasner scalar field solutions (with zero potential, discussed in \Sectionref{Sec:Kasner_solutions_with_a_scalar_field}) to include a potential. We interpret the quantity $A$ as (the analogue of) the scalar field strength since in the special case $\V=0$, $A$ is the scalar field strength, discussed in \Sectionref{Sec:Kasner_solutions_with_a_scalar_field}.
	
	In the special case $A^2=2/3$ we have $\nu^2 = \nu_{\star}^2$. Then, from discussions in \Sectionref{Sec:Asymptotically_Kasner_solutions_in_Bianchi_I_cosmologies}, it follows that $A^2=2/3$ corresponds to a FLRW solution with a non-zero potential. In this situation, it is possible to pick any $\V< 1$. Another interesting case is $A^2=0$, which corresponds to the vacuum case. From \Eqref{ScalarFieldEquation} we see that, if $A=0$, we have $\phi=constant$ and hence $V(\phi(t))=constant$. It therefore follows from \Eqref{Eq:SimplePotential} that we must choose $\V =0$. This makes sense: If the scalar field $\phi$ does not vary in time then the potential $V(\phi)$ must also remain constant. 
	\begin{figure}
		\centering
		\includegraphics[width=0.5\linewidth]{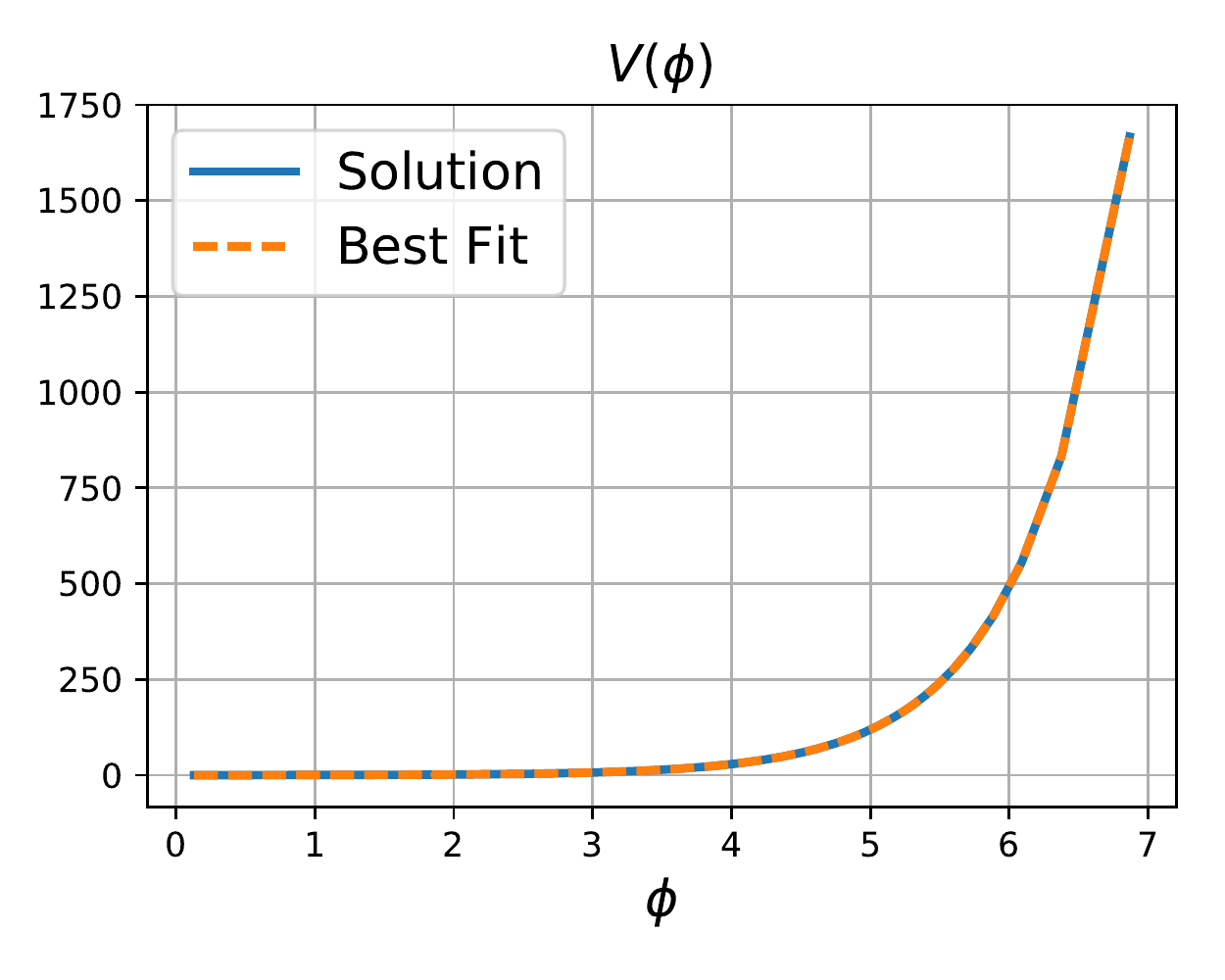}
		\caption{The potential \Eqref{Eq:SimplePotential} as a function of the scalar field (given by \Eqref{ScalarFieldEquation}) for the case $\V = 1/2, A=\sqrt{1/3}$, and $B=0$. Here ``Best Fit'' is given by the first three terms in \Eqref{Eq:SeriesV_Exact}.}
		\label{fig:exactsol}
	\end{figure}
	
	We should remark that although we claim this is a new solution, it is not entirely clear that this is the case. Papers on the subject of Bianchi I solutions with a scalar field seem to always assume a simple form of the potential $V$ as a function of the scalar field. In our case, however, we have a simple functional dependence of $V$ on \textit{time} rather than the scalar field. If we take \Eqref{ScalarFieldEquation} as an implicit formula for $t$ as a function of $\phi$, then $V$ would become a rather involved function of $\phi$. Hence it is likely, although not guaranteed, that this is indeed is a new solution. We can always take \Eqref{ScalarFieldEquation} as an implicit formula for $t$ as a function of $\phi$ as the monotonicity of $\phi$ ensures that the implicit function theorem can always be applied. We now calculate the series expansion of $\phi$ about $t=0$ and invert the series to get $t(\phi)$. Substituting into \Eqref{Eq:SimplePotential} gives
	\begin{align}
	V(\phi)=\frac{\V}{3 A^2} + \frac{\V}{3}e^{ -\frac{\phi-B}{A}} + \frac{(2A^2 -1)\V^3}{4A^4}e^{\frac{\phi-B}{A}}+ O\left( \left(e^{\frac{\phi-B}{A}}\right)^2 \right).
	\label{Eq:SeriesV_Exact}
	\end{align}
	In \Figref{fig:exactsol} we show a plot of $V(t)$ as a function of $\phi(t)$ for the case $\V = 1/2, A=\sqrt{1/3}$ and $B=0$. In the \Figref{fig:exactsol} we also plot the first three non-zero terms given by the series expansion \Eqref{Eq:SeriesV_Exact}.

	\subsection{An exact solution that is not asymptotically Kasner}
	\Thmref{Result:Kanser} gives a clear and easily checkable condition on the potential $V(t)$ that allows one to determine whether or not a (spatially homogeneous) solution is asymptotically Kasner. The goal of the present subsection is to investigate what happens when the Kasner condition \Eqref{Eq:KasnerCondition} does \emph{not} hold. In particular, we search for a strictly monotonic solution (that is not asymptotically Kasner) on the interval $I=\left( 0, T\right]$ with $T=1$. For this, we choose the function
	\begin{align}
	V(t)=\frac{\V-1}{3 \V t^2},\quad \V\ne 0,
	\label{Eq:Potential_NotKasner}
	\end{align}
	where $\V \in \mathbb{R}$ is a freely specifiable (non-zero) constant. It is clear here that $t^2V(t)\rightarrow (\V-1)/(3\V)$ in the limit $t\rightarrow0^+$ and hence \Eqref{Eq:Potential_NotKasner} does \emph{not} satisfy the Kasner condition \Eqref{Eq:KasnerCondition}. We further note that there is no $t\in I$ (or $\V\in\mathbb{R}$) such that $V(t)=1/(3 t^2)$. Using \Eqref{Eq:Lapse_Homogenous} to calculate the lapse $\alpha$ gives
	\begin{align}
	\alpha = \V.
	\label{Eq:Lapse_NoKasnerPot}
	\end{align}
	Requiring that the lapse $\alpha$ is a positive function in the interval $I$ now gives the restriction $\V>0$. It is clear here that if $\V\ne 1$ then the solution is not asymptotically Kasner. Turning our attention to the extrinsic curvature equation \Eqref{Eq:ADM_Homogenous_A} we use \Eqref{Eq:GeneralSolution_A} to find,   
	\begin{align}
	{\chi^a}_{b} &=\exp\left(-\int_{1}^t \frac{\V}{\tau}d\tau + \ln \left( \frac{1}{\sqrt{\V}} {C^a}_b\right) \right)= \frac{{C^a}_b}{\sqrt{\V}t^{\V}},
	\label{Eq:A_NotKasner}
	\end{align}
	where\footnote{Note that the factor of $1/\sqrt{\V}$ has been introduced for later convenience.}  ${\mathcal{C}^{a}}_{b}$ is a symmetric (i.e. $C_{ab}=C_{ba}$) tensor with constant entries subject to the requirement ${\mathcal{C}^{a}}_{a}=0$. Before proceeding we make the further simplification that the extrinsic curvature and metric are both diagonal, i.e,
	\begin{align}
	{\mathcal{C}^a}_b=\text{diag}(\C_1,\C_2,\C_3),\quad \text{and}\quad \gamma_{ab}=0\quad \text{if}\quad a\ne b.
	\end{align}
	As with the previous solution, we point out that this is not a significant restriction as, for spatially homogeneous solutions, one can always perform a (local) coordinate transformation so that ${\chi^a}_{b}$ and $\gamma_{ab}$ are diagonal.

	Inputting \Eqref{Eq:A_NotKasner} into \Eqref{Eq:Evol_y} and solving for the metric components gives
	\begin{align}
	\gamma_{ab}= t^{ \frac{2\V}{3}} \text{diag}\left( \gamma_{1}\exp\left( \frac{2 \sqrt{\V} \C_1  }{(\V -1)t^{\V-1}} \right), \gamma_{2}\exp\left( \frac{2 \sqrt{\V} \C_2  }{(\V -1)t^{\V-1}} \right), \gamma_{3}\exp\left( \frac{2 \sqrt{\V} \C_3  }{(\V -1)t^{\V-1}} \right) \right),
	\label{Eq:ExactMetric_NotKasner}
	\end{align}
	where $\gamma_{1},\gamma_{2}$, and $\gamma_{3}$ are integration constants. In principle these can be any constants, however this just corresponds to a scaling of the spatial coordinates. Notice that in \Eqref{Eq:ExactMetric_NotKasner} we have implicitly restricted our attention to the case $\V\ne 1$. The special case $\V=1$ returns the standard Kasner-scalar field solutions, with zero potential, discussed in \Sectionref{Sec:Kasner_solutions_with_a_scalar_field}.

	We now turn our attention to the scalar field $\phi$. In order to calculate $\phi$ we must first solve for the quantity $\nu^2$. From \Eqref{Eq:General_nu} we find 
	\begin{align}
	\nu^2 = \frac{2}{3 t^2 \V} - \frac{\C}{\V t^{2\V}},
	\end{align}
	where $\C\in\mathbb{R}$ is an integration constant\footnote{The factor of $1/\V$ has been introduced for later convenience.}. In order to ensure that $\nu^2$ is strictly positive on the interval $I = \left( 0, 1 \right]$ we require that the inequality
	\begin{align}
	\frac{2}{3 t^2}-\frac{\C}{t^{2\V}}> 0
	\label{Eq:NotKasner_inequality}
	\end{align}
	holds for all $t\in\left( 0, 1 \right]$. Using \Eqref{Eq:ADM_Homogenous_phi} we now calculate the scalar field as 
	\begin{align}
	\phi=
	\begin{cases}
	B + \sqrt{\frac{2}{3\V}}\ln (t), &\text{if} \quad \C =0,
	\\
	B +\frac{1}{\sqrt{3\V}\left( \V -1 \right)}\left( \phi_1 - \sqrt{ 2 - 3\C t^{ 2( 1 - \V ) } } \right), &\text{if} \quad \C \ne 0, 
	\end{cases}
	\label{ScalarFieldEquation_NotKasner}
	\end{align}
	where, $B\in \mathbb{R}$ is an integration constant and, for the sake of readability, we have defined $\phi_1$ as
	\begin{align}
	\phi_1 = \frac{1}{\sqrt{8}}\left( \ln\left( \sqrt{2}-\sqrt{2-3 \C t^{2(1-\V)}} \right) + \ln\left( \sqrt{2}+\sqrt{2-3 \C t^{2(1-\V)}} \right)  \right).
	\end{align}
	Note here that we have only considered the case $\nu>0$. The case where $\nu<0$ is recovered by the transformation $(\phi,\nu)\mapsto (-\phi,-\nu)$.

	Finally, by appealing to \Thmref{Thm_1}, we note that the Hamiltonian constraint \Eqref{Eq:ADM_Homogenous_Hamiltonian} is satisfied if and only if the algebraic relation 
	\begin{align}
	\C_1^2+\C_2^2+\C_3^2 - {\C} =0 
	\label{Eq:NotKanser_Hamiltonian}
	\end{align}
	holds. Notice that it is possible for \Eqref{Eq:NotKanser_Hamiltonian} to hold only if $\C\ge 0$. Given this, we use the inequality \Eqref{Eq:NotKasner_inequality} to impose a restriction on the possible values of $\V$.

	\begin{figure}[t!]
		\centering
		\includegraphics[width=0.5\linewidth]{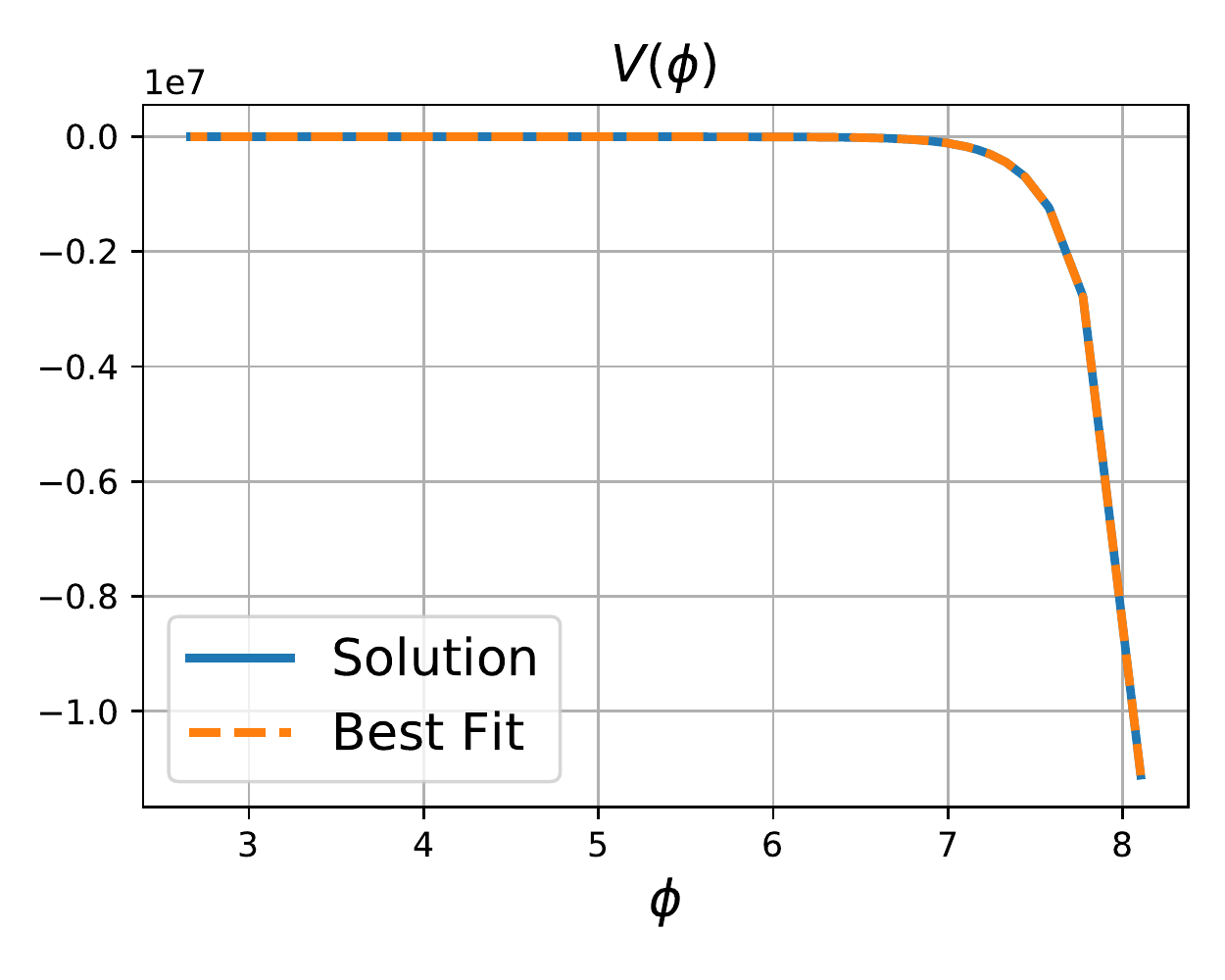}
		\caption{The potential \Eqref{Eq:Potential_NotKasner} as a function of the scalar field (given by \Eqref{ScalarFieldEquation_NotKasner}) for the case $\V = 3/4, \C = {1/3}$ and $B=0$. Here, ``Best Fit'' is given by the first three terms in \Eqref{Eq:V_Expand_notKasner}.}
		\label{fig:exactsolnotkasner}
	\end{figure}
	If the inequality \Eqref{Eq:NotKasner_inequality} holds at $t=1$, we find that $\C < 2/3$. It follows then that the largest $\C$ can be is $\C\approx 2/3$. In this extremal case we further find that \Eqref{Eq:NotKasner_inequality} holds if and only if $\V \in \left( 0, 1 \right)$. We therefore have the following requirements:
	\begin{align}
	\begin{cases}
	0< \C < \frac{2}{3},
	\\
	0<\V <1.
	\end{cases}
	\label{Eq:Inequalities_NotKasner}
	\end{align}
	In the special case $\C =0$ we find that $\V>0$ is the only restriction we need to make on $\V$. Note that, unlike the solution presented in \Sectionref{SubSec:An_exact_Kasner-like_solution_with_a_singular_potential} there is no value of $\C$ that would return the vacuum Kasner scalar field solutions with zero potential, presented in \Sectionref{Sec:Kasner_solutions_with_a_scalar_field}. This is a direct consequence of the exclusion of the $\V=1$ case.

	Provided the inequalities \Eqref{Eq:Inequalities_NotKasner} hold, we have that $\partial_{t}\phi\ne 0$ for all $t\in I$. In particular, the implicit function theorem implies that it is possible to (at least locally) invert $\phi(t)$ to find $t(\phi)$ and hence $V(\phi)=V(t(\phi))$. We now calculate the series expansion of $\phi$ about $t=0$ and invert the series to get $t(\phi)$. Substituting into \Eqref{Eq:Potential_NotKasner} now gives
	\begin{align}
	V(\phi)=3^{\frac{1}{1-\V}}\text{e}^{\frac{4}{\V-1}}\left( 1 - \frac{\text{e}^{{w(\phi)}}}{\V-1} + \frac{\text{e}^{8}\left( 7\V - 3 \right)\text{e}^{{2w(\phi)}}}{8(\V-1)^2} + O\left( \left(\text{e}^{{w(\phi)}} \right)^3 \right) \right)\text{e}^{ \frac{w(\phi)}{\V-1} }, 
	\label{Eq:V_Expand_notKasner}
	\end{align}
	where we have defined 
	\begin{align*}
	w(\phi)=2\sqrt{6\V}(\V-1)\left( \phi - B \right).
	\end{align*}
	In \Figref{fig:exactsolnotkasner} we show a plot of $V(t)$ as a function of $\phi(t)$ for the case $\V = 3/4, \C = {1/3}$ and $B=0$. In the \Figref{fig:exactsolnotkasner} we also plot the first three non-zero terms given by the series expansion \Eqref{Eq:V_Expand_notKasner}.

	\section{Numerical examples}
	\label{Numerical_examples}
	\subsection{Numerical set-up}
	\subsubsection{Potential and initial conditions}
	In previous sections we have provided an analytical theory that describes the effect a potential $V(t)$ has on Bianchi I solutions with a strictly monotonic scalar field $\phi$. The goal of the present section is to numerically extend these results and to provide some (numerical) examples (and expansion formulas) of the various types of spatially homogeneous space-times, discussed at the end of \Sectionref{Sec:Asymptotically_Kasner_solutions_in_Bianchi_I_cosmologies}. Namely, we provide numerical examples for which the scalar field $\phi$ is (1) strictly monotonic, (2) asymptotically stationary, and (3) eventually-monotonic. To this end we seek to solve the Einstein scalar field equations \Eqsref{Eq:ADM_Homogenous_A}--\eqref{Eq:Lapse_Homogenous} on the interval $I=\left( 0, 1 \right]$. Although our Python code allows for an arbitrary choice of potential $V(\phi)$, we restrict our attention to the choice 
	\begin{align}
	V(\phi) = m\left( \cosh\left( k\phi \right) - 1 \right),
	\label{Eq:CoshPotential}
	\end{align}
	where $m,k\in\mathbb{R}$ are positive freely specifiable constants. At this stage it is not clear that \Eqref{Eq:CoshPotential} allows for each of the various types of space-times. However, we find that this potential does allow for each type of solution. Observe carefully that, unlike \Sectionref{BianchiI}, we specify the potential $V(\phi)$ in terms of the \emph{scalar field} $\phi$ and \emph{not as a function of time}.
	
	Let us now discuss how we construct our initial data. Recall first that an initial data set $(\gamma_{ab},K_{ab})$ must satisfy the Hamiltonian constraint \Eqref{Eq:ADM_Homogenous_Hamiltonian} at the initial time. For this we first pick a \emph{background initial data set} $(\mathring{\gamma}_{ab},\mathring{\chi}^a{}_b,\mathring{\phi})$ satisfying the Hamiltonian constraint. Here, we pick the background to be an exact Kasner scalar field solution (see \Sectionref{Sec:Kasner_solutions_with_a_scalar_field}). At the initial time $t=1$ we set
	\begin{align}
	\gamma_{ab}(t=1)= \mathring{\gamma}_{ab},
	\quad
	{\chi}^a{}_b(t=1)=\mathring{\chi}^a{}_b,
	\quad
	\phi(t=1)=0.
	\label{Eq:Cosh_ID}
	\end{align} 
	We find that this choice of initial data ensures that the Hamiltonian constraint \Eqref{Eq:ADM_Homogenous_Hamiltonian} is satisfied initially provided $\nu(t=1)=\mathring{A}$, where $\mathring{A}$ is the scalar field strength associated with the background data set $(\mathring{\gamma}_{ab},{\mathring{A}^a}{}_b,\mathring{\phi})$.  Given this initial data we then numerical evolve the unknowns towards the singularity at $t=0$. For our time-stepping method we use the adaptive \emph{SciPy} integrator \emph{odeint}\footnote{See \url{https://docs.scipy.org/doc/scipy/reference/generated/scipy.integrate.odeint.html}.}. We denote its absolute error as $\tilde{\mathcal{E}}$ and its magnitude as $\tilde{\varepsilon}=-\log(\tilde{\mathcal{E}})$. It is clear then that $\tilde{\varepsilon}$ controls the local step size of the time evolution and is the primary source of error in our numerical scheme, and should decrease monotonically with $\tilde{\varepsilon}$. Throughout our evolution we use the relative constraint violation 
	\begin{align}
	C = \left| t^2 \chi^{a}{}_{b}\chi_{a}{}^{b} + t^2 \nu^2 + 2t^2 V(\phi) -\frac{2}{3}  \right|
	\end{align}
	as a gauge of our numerical error.

	\subsubsection{Perturbation expansions}
	Recall now that at $t=1$ we have $\phi=0$. In particular this means that $V(\phi)$, as given by \Eqref{Eq:CoshPotential}, vanishes initially. With this in mind, we interpret the solutions we construct here as spatially homogeneous perturbations of the exact Kasner scalar field solutions with zero potential $(\mathring{\gamma}_{ab},{\mathring{A}^a}{}_b,\mathring{\phi})$. It therefore makes sense to consider perturbation expansions of $\phi$ with respect to the parameters $m$ and $k$. In particular we expect that if $k$ (or $m$) is `\emph{small}' then we can write $\phi$ as a series of the form 
	\begin{align}
	\phi = \phi_{(0)} + \sum_{i=1}^{\infty}\vartheta^{i}\phi_{(i)},
	\label{Eq:PerturbationExpansion_Generic}
	\end{align}
	where we have that either $\vartheta=k$ or $\vartheta =m$ and the $\phi_{(i)}$'s are unknown functions of time. From \Eqref{Eq:Cosh_ID} we find that the fields $\phi_{(i)}$ satisfy the initial conditions 
	\begin{align}
	\phi_{(0)}(1)=\phi_{(i)}(1)=0, \quad \partial_{t}\phi_{(0)}(1)=\mathring{A},\quad \partial_t \phi_{(i)}(1)=0,\quad i=1,\dots, \infty.
	\label{Eq:ID_PerturbExpand}
	\end{align}
	Inputting the expansion \Eqref{Eq:PerturbationExpansion_Generic} into \Eqref{Eq:ADM_Homogenous_phi} we find that, irrespective of $\vartheta=k,m$, the function $\phi_{(0)}$ must satisfy the differential equation 
	\begin{align}
	\partial_{t}^2 \phi_{(0)}+\frac{1}{t}\partial_t \phi_{(0)} =0, \quad \phi_{(0)}(1)=0, \quad \partial_t \phi_{(0)}(1)=\mathring{A}.
	\end{align}
	It follows then that 
	\begin{align}
	\phi_{(0)} = \mathring{A}\ln(t).
	\label{Eq:PerturbationExpansion_phi0}
	\end{align} 
	Generically, one expects each of the $\phi_{(i)}$'s to become infinite in the limit $t\rightarrow 0^+$, and hence the term $\vartheta^i \phi_{(i)}$ is not `small', for $t$ sufficiently close to $t=0$. It is therefore worth pointing out that when we say $\vartheta$ is \emph{small} we mean that $|\vartheta^i \phi_{(i)}/\phi_{(0)}| \ll 1$ for all $t\in \left( 0, 1\right)$. Of course, as a matter of principle, one cannot a priori determine the validity of such an assumption. Instead, one must first determine the function $\phi_{(i)}$ and then check for consistency. Explicit formulas for the remaining $\phi_{(i)}$'s are given in the following sections.

	Finally, we end this subsection by introducing the function $\K(t)$ as 
	\begin{align}
	{\K(t)}=t^2\nu^2 + t^2 {\chi^a}_{b}{\chi^b}_{a}- \frac{2}{3}.
	\label{Eq:K_Def}
	\end{align} 
	We find that if $\K(t)\rightarrow 1$ as $t\rightarrow 0^+$ then the corresponding solution is asymptotically Kasner in the sense of \Defref{Def:AsymptoticallyKasner}. We therefore refer to $\K(t)$ as \emph{`the Kasner constraint'} and $|\K(t)-1|$ as the \emph{`violation of the Kasner constraint'}. Note here that we say a (spatially homogenous) solution is asymptotically Kasner if the final value of $|\K-1|$ is smaller that the constraint violation.

	\subsection[Expansions in $m$ and their numerical justifications]{Perturbation expansions in $m$ and their numerical justifications}
	\label{SubSec:Solutions_that_are_not_asymptotically_Kasner}
	In this subsection we consider the behaviour of the scalar field $\phi$, corresponding to the potential \Eqref{Eq:CoshPotential}, when the parameter $m$ is small. It is clear from \Eqref{Eq:CoshPotential} that $m$ only affects the relative size of the potential $V(\phi)$. As such, we do not expect $m$ to have any impact on the generic time dependence of the potential $V(\phi(t))$. It therefore makes sense to establish the general behaviour of the scalar field by considering a perturbation expansion of $\phi$ in $m$. i.e. we assume that the scalar field $\phi$ has the asymptotic form
	\begin{align}
	\phi= \mathring{A}\ln(t)+m\phi_{(1)} + O\left(m^2\right),
	\label{Eq:PerturbationExpansion_m}
	\end{align}
	where $\phi_{(1)}$ is a time-dependent function. Observe that \Eqref{Eq:PerturbationExpansion_m} is consistent with \Eqref{Eq:PerturbationExpansion_Generic} and \Eqref{Eq:PerturbationExpansion_phi0} for $\vartheta=m$. Inputting the expansion \Eqref{Eq:PerturbationExpansion_m} into the (spatially homogeneous) scalar field equation \Eqref{Eq:ADM_Homogenous_phi} and using that the initial conditions are given by \Eqref{Eq:ID_PerturbExpand}, we find that $\phi_{(1)}$ must satisfy the differential equation
	\begin{align}
	\begin{cases}
	\partial_{t}^2 \phi_{(1)}+\frac{1}{t}\partial_t \phi_{(1)}  =-\frac{3\mathring{A}}{2}\left( t^{k\mathring{A}} + t^{-k\mathring{A}} \right)-\frac{(3\mathring{A}^2 - 1)k}{2}\left( t^{k\mathring{A}} - t^{-k\mathring{A}} \right), 
	\\
	\phi_{(1)}(1)=0,\quad \partial_{t}\phi_{(1)}(1)=0,
	\end{cases}
	\end{align}
	and therefore has solution
	\begin{align}
	\phi_{(1)}=
	\begin{cases}
	\frac{\mathring{A} (k^2 ( 12 \mathring{A}^2 + 3 \mathring{A}^4 k^2 - 16) - 3 (\mathring{A}^2 k^2 -4)^2 t^2 - 
		2 (3 \mathring{A}^2-2) (\mathring{A}^2 k^2 -4) \ln(t))k^2 + {\varphi}_{(1)}}{4 (\mathring{A}^2 k^2 - 4)^2}, &k\ne {2}/{\mathring{A}},
	\\
	\frac{2 - 2 t^4 + 3 \mathring{A}^2 (5 - 8 t^2 + 3 t^4) + 
		4 \ln(t) (2 + 3 \mathring{A}^2 + (4 - 6 \mathring{A}^2) \ln(t))}{32 \mathring{A}},  &k={2}/{\mathring{A}},
	\end{cases}
	\label{Eq:phi1_mExapansion}
	\end{align}
	where
	\begin{align}
	\varphi_{(1)}=  2 ( \mathring{A} k -2)^2 (3 \mathring{A} (1 + \mathring{A} k)-k) t^{2 + \mathring{A} k} -2 (2 + \mathring{A} k)^2 (3 \mathring{A} (\mathring{A} k -1) -k) t^{2 - \mathring{A} k}.
	\end{align}
	Using \Eqref{Eq:PerturbationExpansion_m} and \Eqref{Eq:phi1_mExapansion} we can now calculate the leading order behaviour of the potential \Eqref{Eq:CoshPotential} as a function of time. We find,
	\begin{align}
	V(\phi(t))=\frac{1}{2}\left( t^{\mathring{A} k} + t^{-\mathring{A} k} - 2 \right)m + O(m^2).
	\label{Eq:PotentialExpand_m}
	\end{align}
	From this expansion one might conclude that \emph{$t^2 V(\phi)\rightarrow 0$ as $t\rightarrow 0^+$ if and only if $| \mathring{A} k |< 2$}. If the inequality $| \mathring{A} k |< 2$ is satisfied then we find that \Thmref{Result:Kanser} holds for any $\epsilon\in ( 0, 2-|\mathring{A}k| )$. It is worth noting here that the assumption that $m$ is small (as was discussed in \Sectionref{SubSec:Potential_initial_conditions_and_perturbation_expansions}) holds if $mk^2\ll1$.
	
	\begin{figure}
		\centering
		\includegraphics[width=1.0\linewidth]{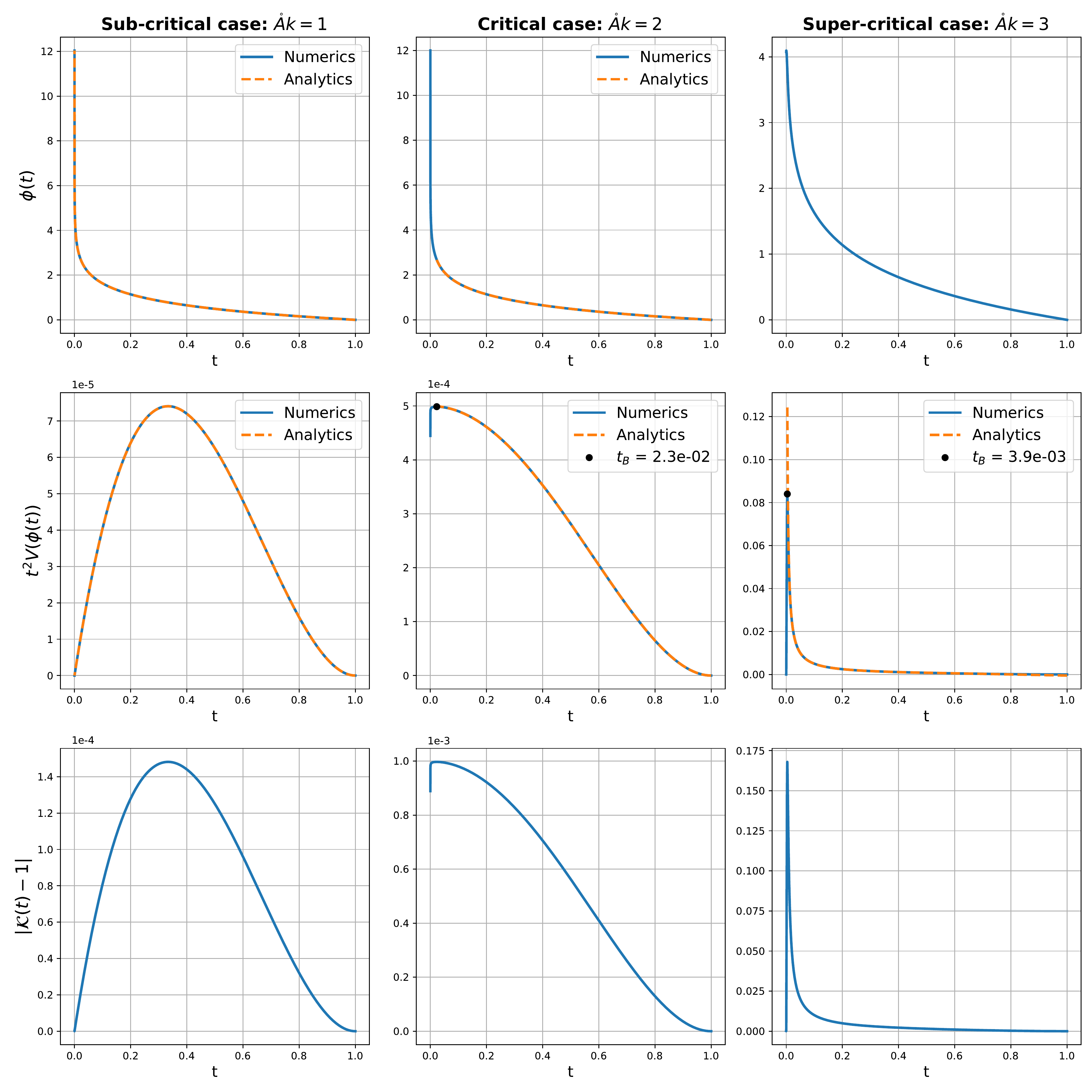}
		\caption{The numerical solutions corresponding to $m=10^{-3},\theta=-\pi/3$ and $\psi=5\pi/3$ so that $\mathring{A}=1/\sqrt{2}$. Each of the three columns correspond to the choices $k=\sqrt{2},2\sqrt{2}$ and $3\sqrt{2}$, respectively.}
		\label{fig:mexpansiontests}
	\end{figure}
	\begin{figure}[t!]
		\centering
		\includegraphics[width=0.4\linewidth]{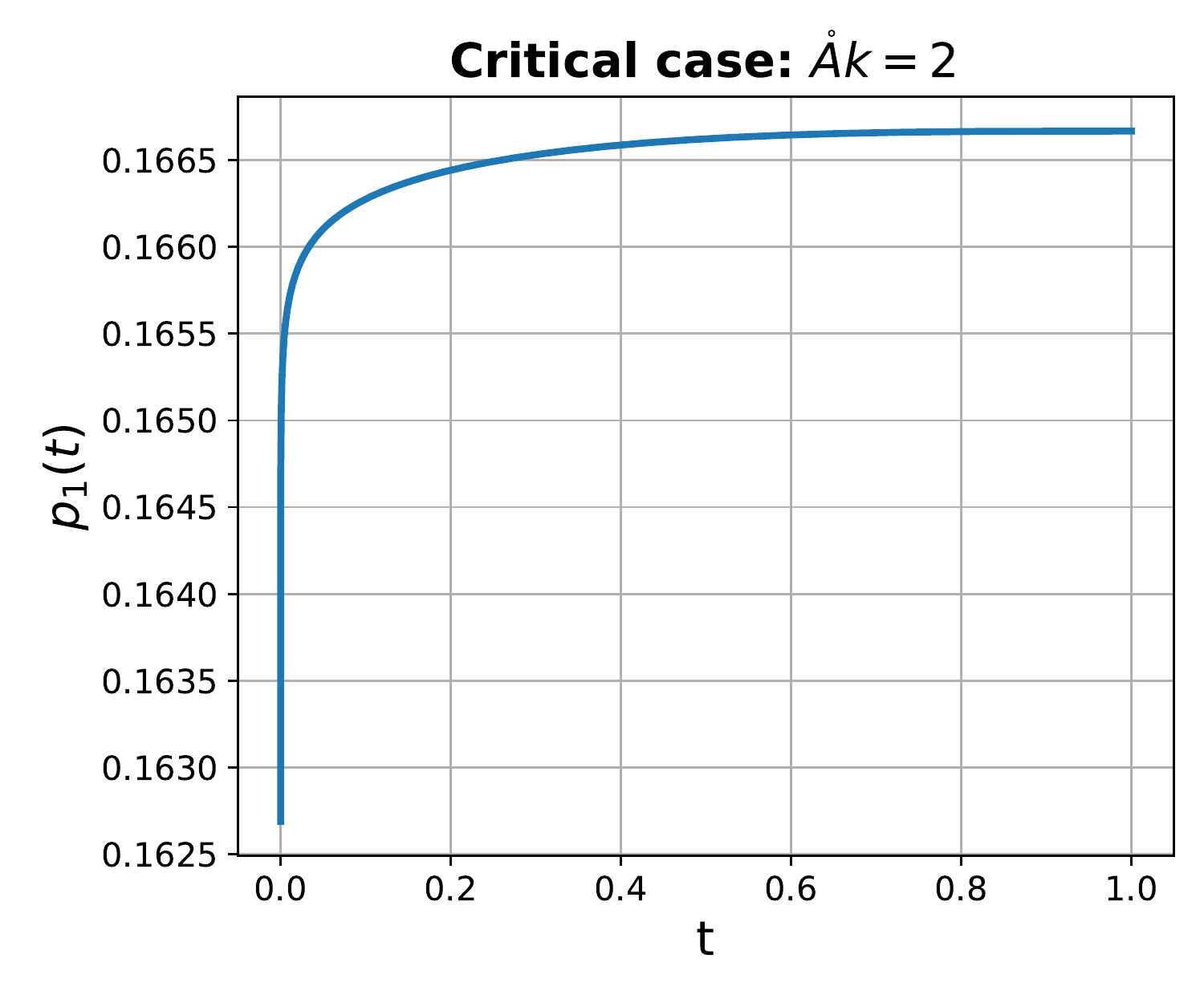}
		\includegraphics[width=0.4\linewidth]{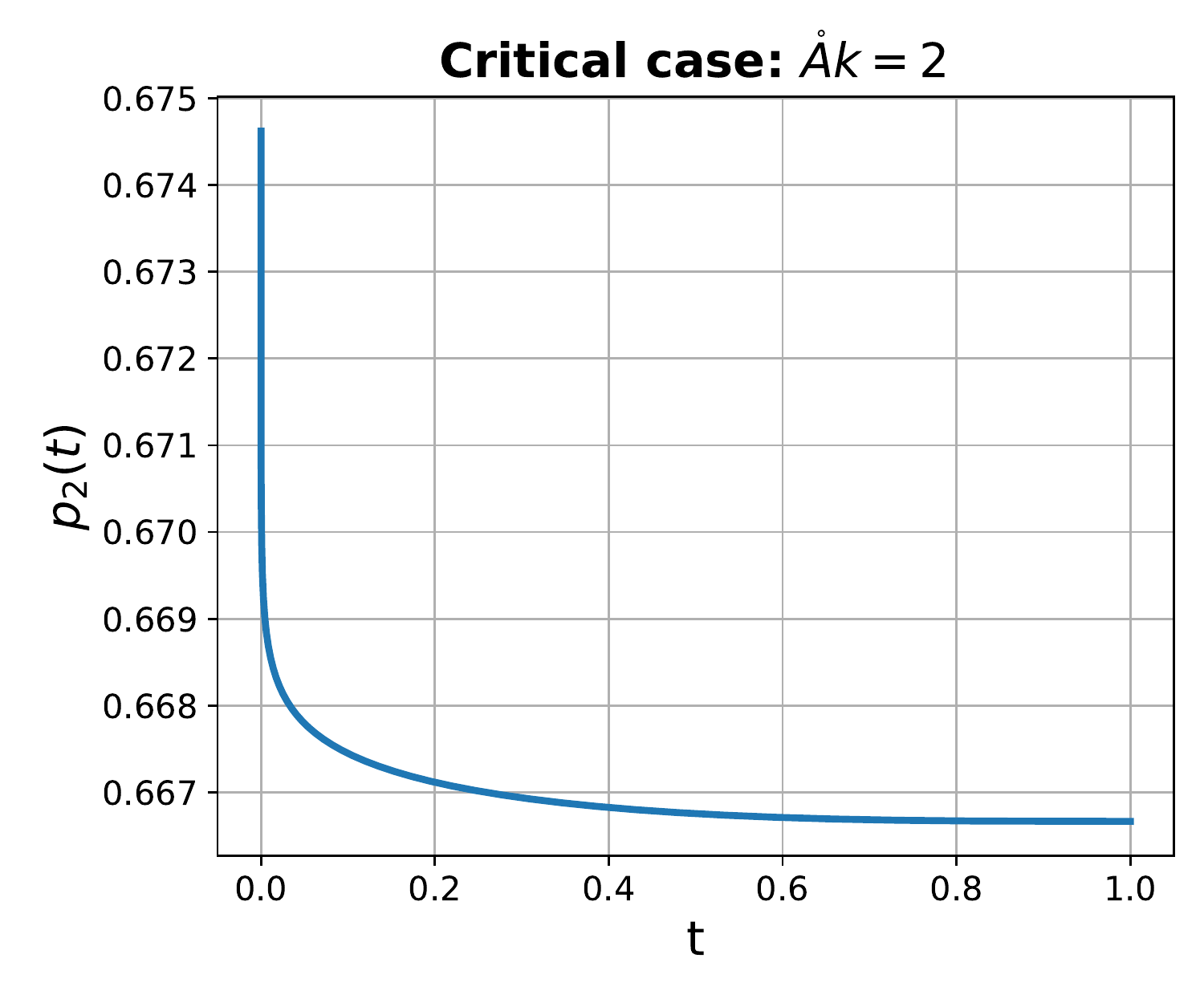} \\
		\includegraphics[width=0.4\linewidth]{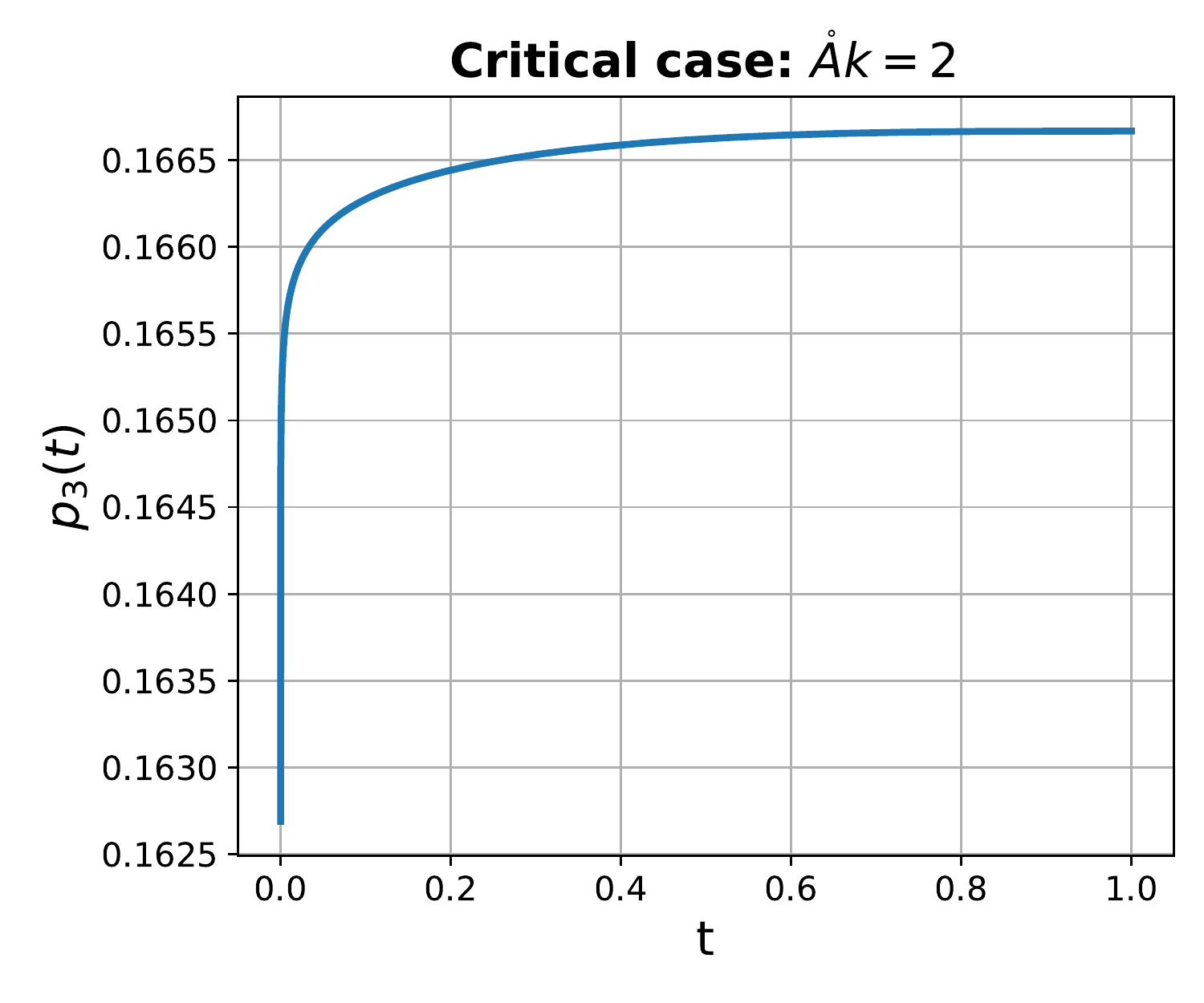}
		\includegraphics[width=0.4\linewidth]{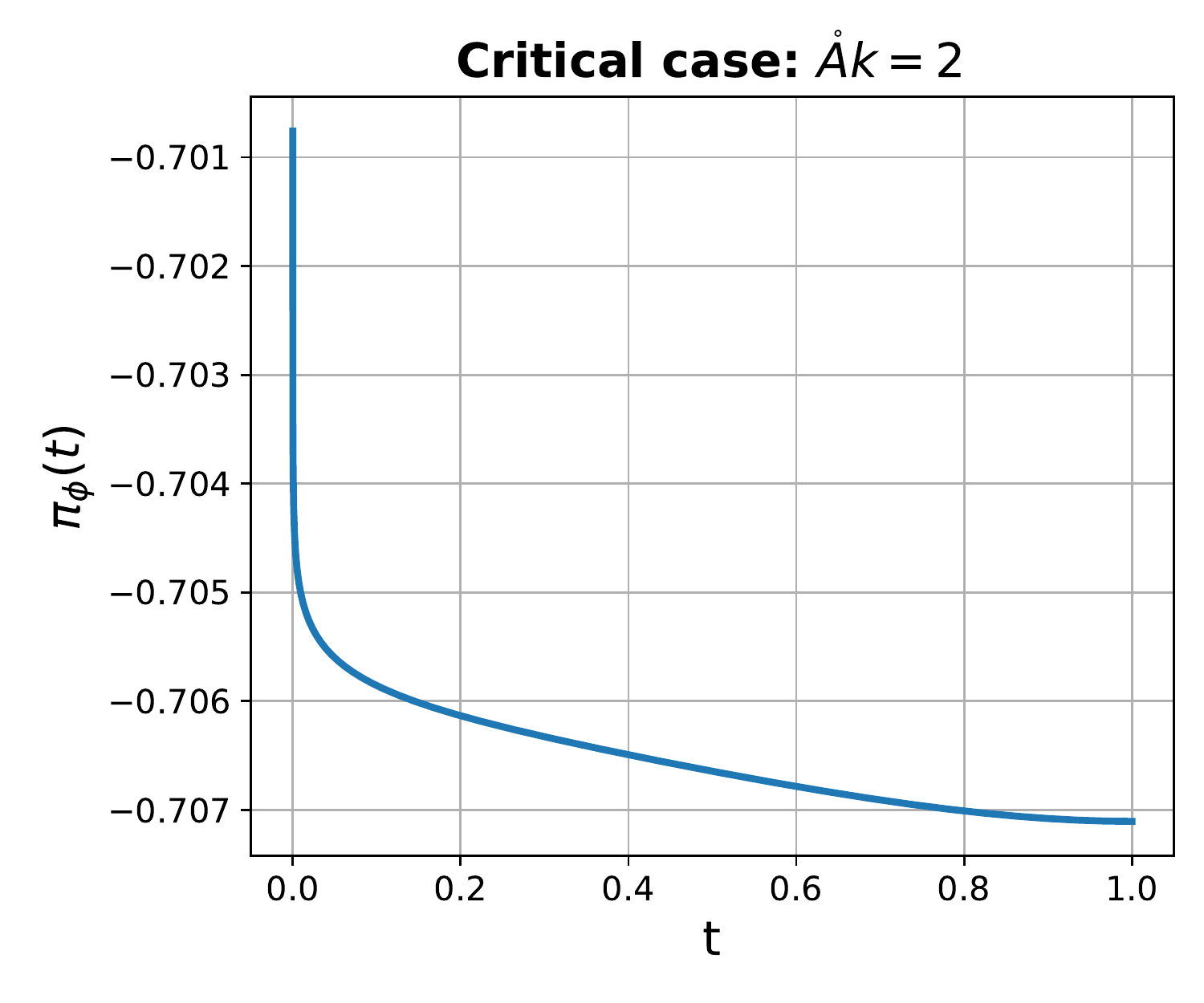}
		\caption{The numerically calculated Kasner exponents and scalar field strength correspond\-ing to $m=10^{-3},\theta=-\pi/3,\psi=5\pi/3$ (so that $\mathring{A}=1/\sqrt{2}$) and $k=2/\mathring{A}$. }
		\label{fig:p1bouncecrt}
	\end{figure}
	\begin{figure}[t!]
		\centering
		\includegraphics[width=0.4\linewidth]{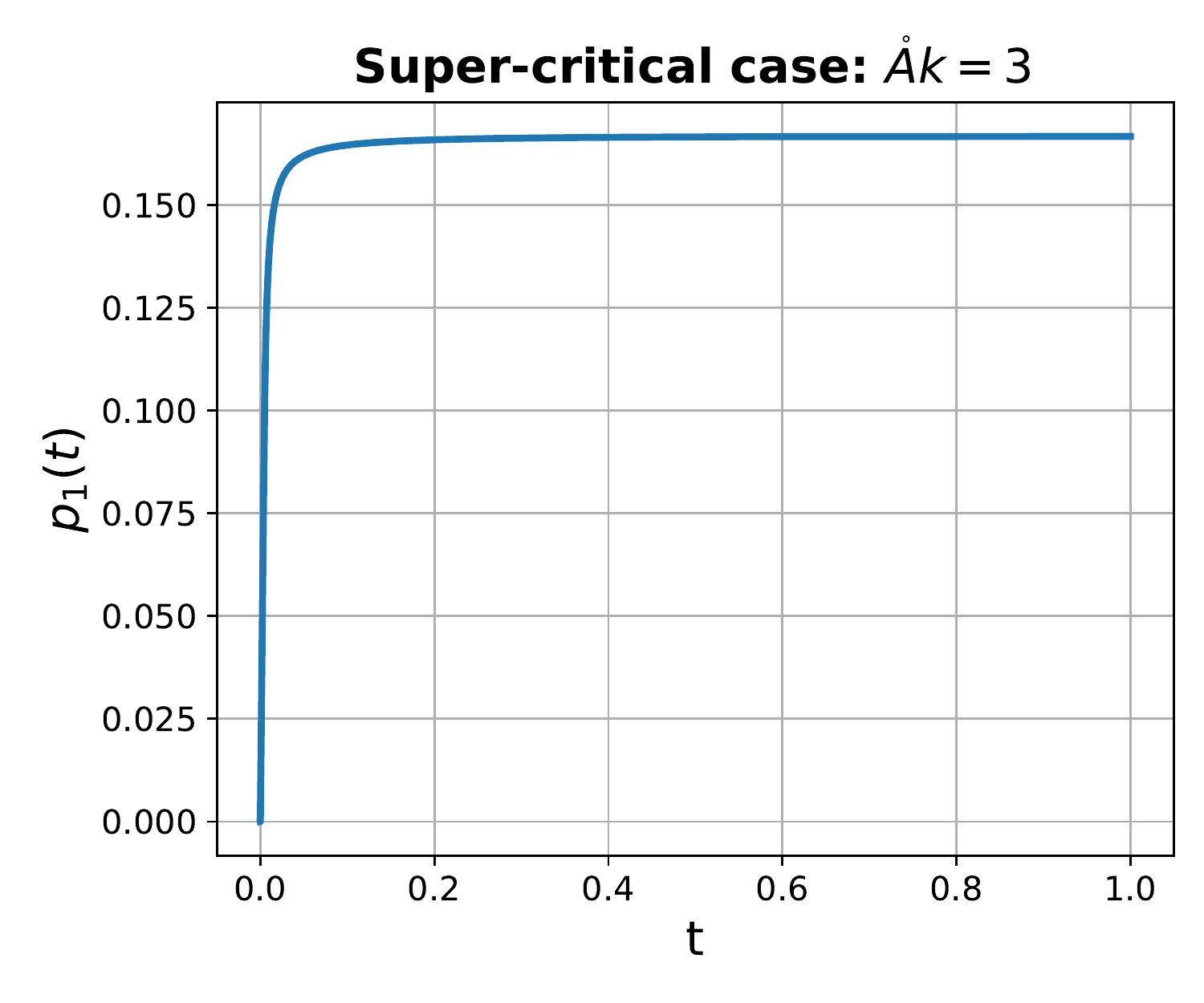}
		\includegraphics[width=0.4\linewidth]{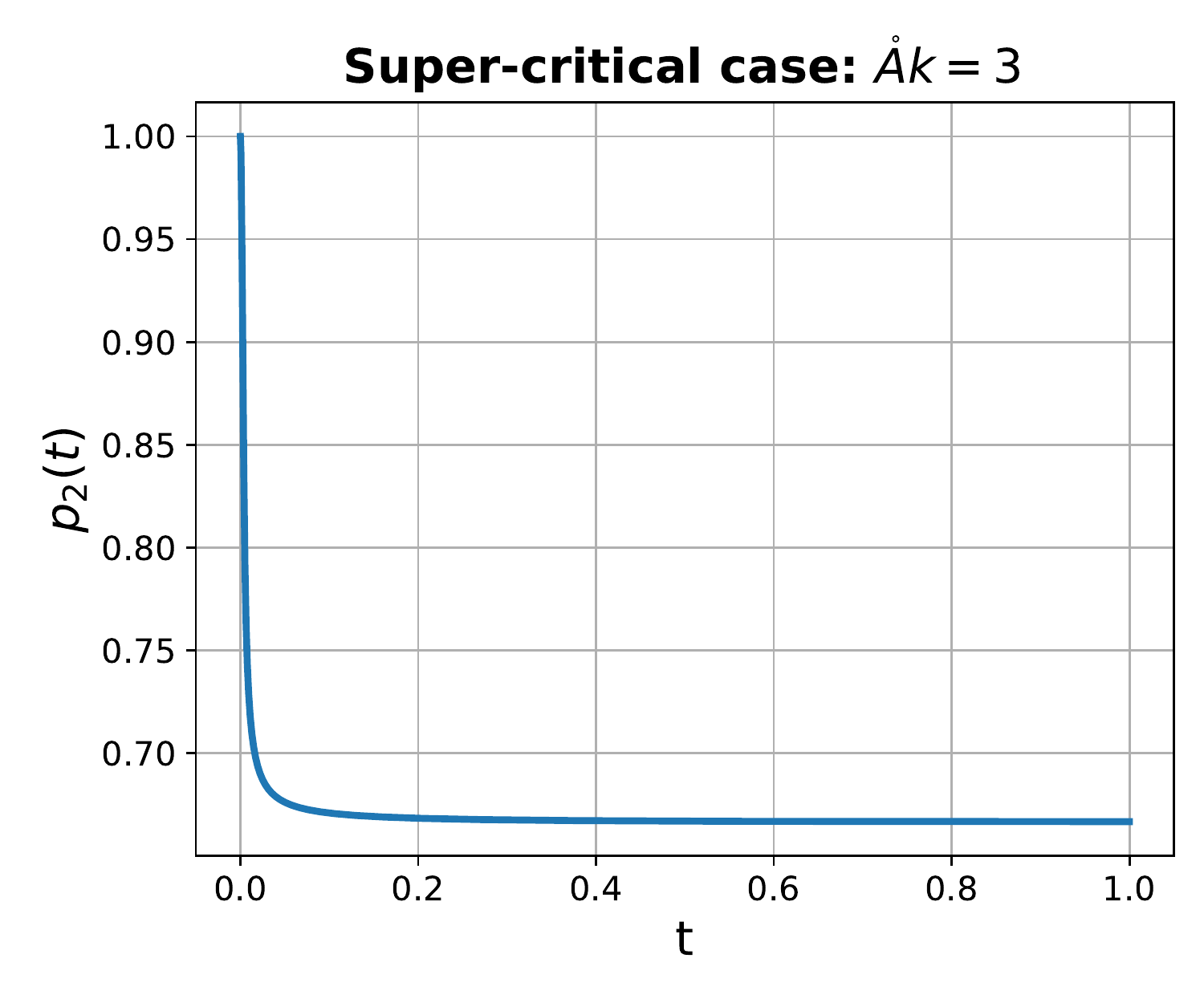} \\
		\includegraphics[width=0.4\linewidth]{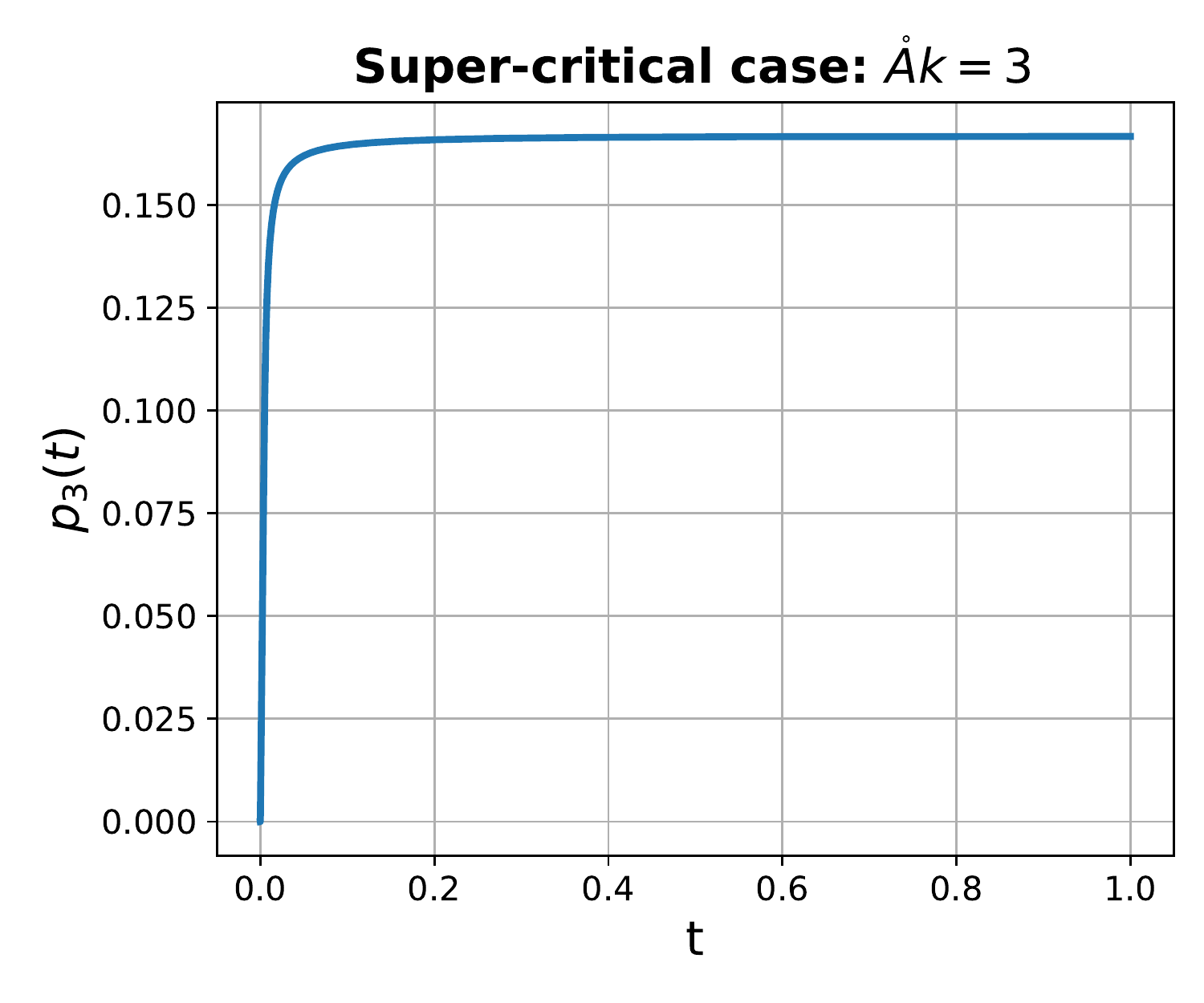}
		\includegraphics[width=0.4\linewidth]{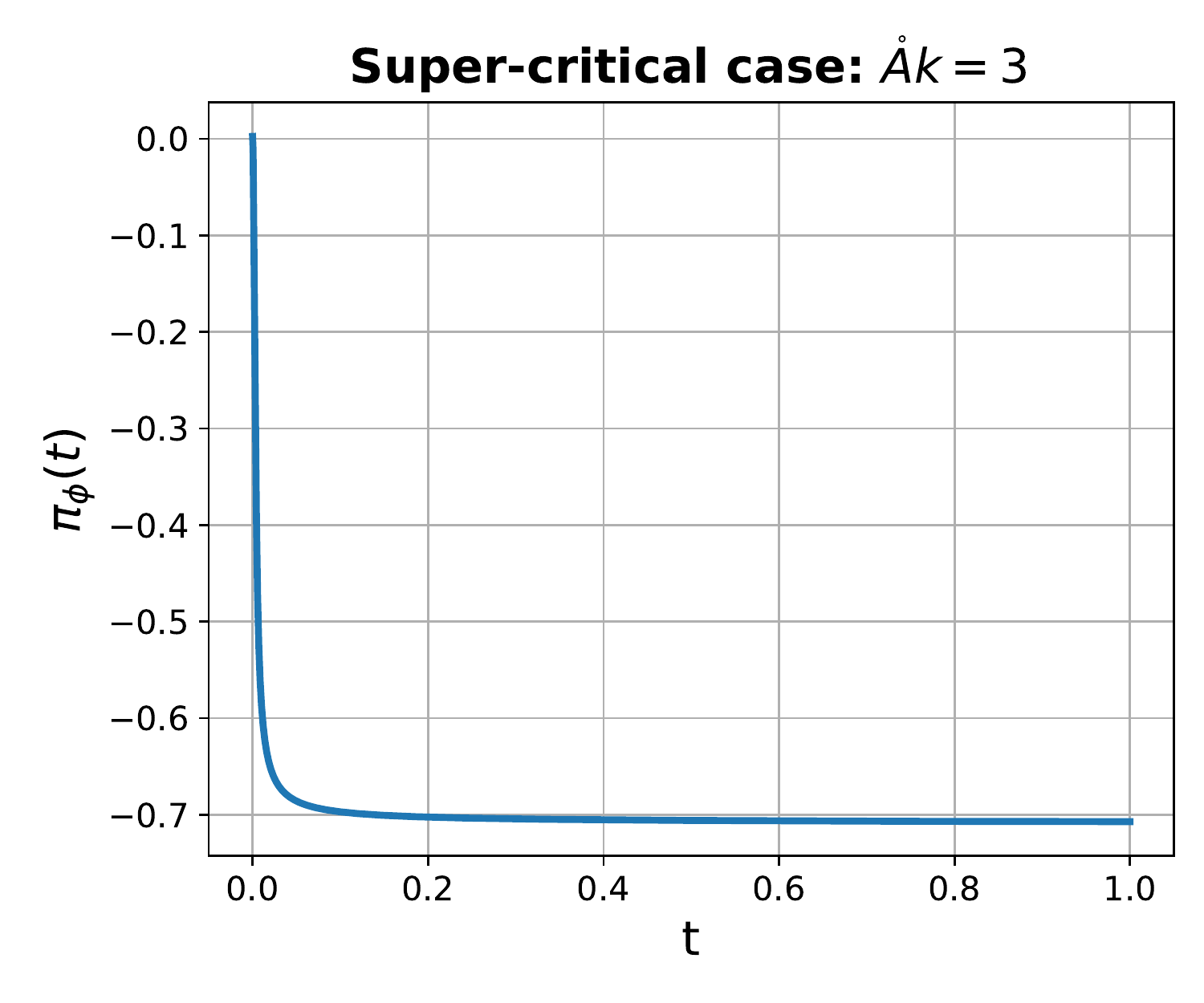}
		\caption{The numerically calculated Kasner exponents and scalar field strength correspond\-ing to $m=10^{-3},\theta=-\pi/3,\psi=5\pi/3$ (so that $\mathring{A}=1/\sqrt{2}$) and $k=3/\mathring{A}$. }
		\label{fig:p1bouncesup}
	\end{figure}

	We now provide numerical support for the expansions \Eqref{Eq:PerturbationExpansion_m}--\eqref{Eq:PotentialExpand_m}. In \Figref{fig:mexpansiontests} we show the results for three of our numerical simulations in which the value of $k$ changes but all other parameters are kept fixed. For each of the plots in \Figref{fig:mexpansiontests} we set $m=10^{-3},\theta=-\pi/3$ and $\psi=5\pi/3$. In the first column, we show the solutions corresponding to the choice $k=\sqrt{2}$ so that $|\mathring{A}k|=1$. In the second column, we show the solutions corresponding to the choice $k=2\sqrt{2}$ so that $|\mathring{A}k|=2$ and in the third column, we show the solutions corresponding to the choice $k=3\sqrt{2}$ so that $|\mathring{A}k|=3$. In the top row of \Figref{fig:mexpansiontests} we show the scalar field $\phi$. In the second and third rows of \Figref{fig:mexpansiontests} we show the quantities $t^{2} V(\phi(t))$ and $|\K(t)-1|$, respectively. For the `sub-critical' case, for which we have $\mathring{A}k=1$, we see that our analytical predictions closely match our numerical results. However, in the `critical' and `super-critical' cases we find that our analytical predictions \emph{do not} match our numerical simulations. In the critical case we find that the numerical solutions closely match the analytical predictions for $t\in [t_B,1]$. However, at $t=t_B\approx 2.3\times 10^{-2}$ the solution \emph{bounces} to a space-time that \emph{is} asymptotically Kasner. This behaviour becomes more extreme in the super-critical case where we see, in the centre right plot of \Figref{fig:mexpansiontests}, that although the analytical predictions do match the numerical results for a short while, by the time of the \emph{bounce} at $t=t_B\approx 3.9\times 10^{-3}$ the analytical prediction for $t^2 V(\phi(t))$ is significantly larger than the numerically calculated value of $t^2 V(\phi(t))$. It is not entirely surprising that this bounce type behaviour is not predicted by \Eqref{Eq:PerturbationExpansion_m}--\eqref{Eq:PotentialExpand_m} as this type of bouncing phenomena is a highly non-linear process and hence cannot be approximated by the linearisation. Observe carefully that this bounce is clearly seen in the ``super-critical case''. However, in the ``critical case'' we only see the beginning of the bounce. It is therefore not clear that this solution does indeed bounce to a solution that is asymptotically Kasner. The only way to demonstrate that the solution does indeed become asymptotically Kasner would be to re-perform the simulations closer to $t=0$. However, we find that our code is not able to get closer to the singularity than $t=4\times 10^{-8}$.  
	
	Finally, in \Figref{fig:p1bouncecrt} we show the Kasner coefficients $p_{1}(t),p_{2}(t),p_{3}(t)$ and the conjugate momentum of the scalar field\footnote{Note that the scalar field strength $A$ is calculated as $\lim_{t\rightarrow 0^{+}}\pi_{\phi}\phi(t)$} $\pi_{\phi}(t):=t\partial_{t}\phi/\alpha$, in both the critical and super-critical cases. In each of the plots in \Figref{fig:p1bouncecrt} we see that the quantities are approximately constant before \emph{bouncing} to a different value.

	\subsection[Expansions in $k$ and their numerical justifications]{Perturbation expansions in $k$ and their numerical justifications}
	\label{SubSec:Perturbation_expansions_in_k}
	\subsubsection{Taylor expansions in $k$}
	The goal of the present subsection is to provide a detailed description of solutions for which the scalar field $\phi$ is a strictly monotonic function, and for which the Kasner condition \Eqref{Eq:KasnerCondition} holds. We saw in the previous subsection that solutions (corresponding to the potential \Eqref{Eq:CoshPotential}) were asymptotically Kasner if and only if $|\mathring{A}k|<2$. Clearly, this inequality always holds if $k$ is small (since, from \Sectionref{Sec:Kasner_solutions_with_a_scalar_field}, we know that $|\mathring{A}|<\sqrt{2/3}$). In such a setting, one expects the scalar field $\phi$ to have an asymptotic expansion of the form
	\begin{align}
	\phi=\mathring{A}\ln(t) + \phi_{(1)}k + \phi_{(2)}k^2 + \phi_{(3)}k^3 + O(k^4), 
	\label{Eq:PerturbationExpansion_k}
	\end{align}
	where $\phi_{(1)},\phi_{(2)},\phi_{(3)}$ are time-dependent functions. Observe that \Eqref{Eq:PerturbationExpansion_k} is consistent with \Eqref{Eq:PerturbationExpansion_Generic} and \Eqref{Eq:PerturbationExpansion_phi0} for $\vartheta=k$. Inputting \Eqref{Eq:PerturbationExpansion_k} into the scalar field equation \Eqref{Eq:ADM_Homogenous_phi}, we find that if $i$ is an odd number then $\phi_{(i)}$ is a solution of the equation
	\begin{align}
	\frac{d^2 \phi_{(i)}}{dt^2}+\frac{1}{t}\frac{d\phi_{(i)}}{dt}=0, \quad \phi_{(i)}(1)=0,\quad \frac{d\phi_{(i)}}{dt}(1)=0.
	\end{align}
	It follows then that $\phi_{(i)}=0$ if $i$ is odd. In particular we have $\phi_{(1)}=\phi_{(3)}=0$ and hence $\phi_{(2)}$ is the only non-zero function in \Eqref{Eq:PerturbationExpansion_k}. It now follows from \Eqref{Eq:ADM_Homogenous_phi} that $\phi_{(2)}$ is a solution of the differential equation 
	\begin{align}
	\partial_t^{2} {\phi_{(2)}}+\frac{1}{t}\partial_{t}{\phi_{(2)}}=\frac{1}{2} \mathring{A} m \ln (t) \left(3 \mathring{A}^2 \ln (t)+6 \mathring{A}^2-2\right), \;\; \phi_{(2)}(1)=0, \;\; \partial_t {\phi_{(2)}}(1)=0,
	\end{align}
	and is therefore
	\begin{align}
	\phi_{(2)}=\frac{1}{16} \mathring{A} m \left(\left(4-3 \mathring{A}^2\right) \left(t^2-1\right)+2 \ln (t) \left(3 \mathring{A}^2 t^2 \ln (t)+3 \mathring{A}^2-2 \left(t^2+1\right)\right)\right).
	\label{Eq:phi2_kExapansion}
	\end{align}
	Observe here that the assumption that $k$ is small holds if $mk^2 \ll 1$. If $mk^2 \ll 1$ does not hold then it is possible that the scalar field $\phi$ is not a strictly monotonic function. Using \Eqref{Eq:PerturbationExpansion_k} and \Eqref{Eq:phi2_kExapansion} we can now calculate the leading order behaviour of the potential \Eqref{Eq:CoshPotential} as a function of time. We find,
	\begin{align}
	V(\phi(t))=\frac{1}{2}\mathring{A}^2 m \ln(t)^2 k^2+ O(k^4).
	\label{Eq:PotentialExpand_k}
	\end{align} 
	It is clear from \Eqref{Eq:PotentialExpand_k} (provided $mk^2\ll1$) that $t^2V(\phi)\rightarrow 0$ in the limit $t\rightarrow 0^+$ and in particular we find that \Thmref{Result:Kanser} holds for any $\epsilon\in (0,2)$, as was expected. 
	\begin{figure}[t]
		\centering
		\includegraphics[width=1.0\linewidth]{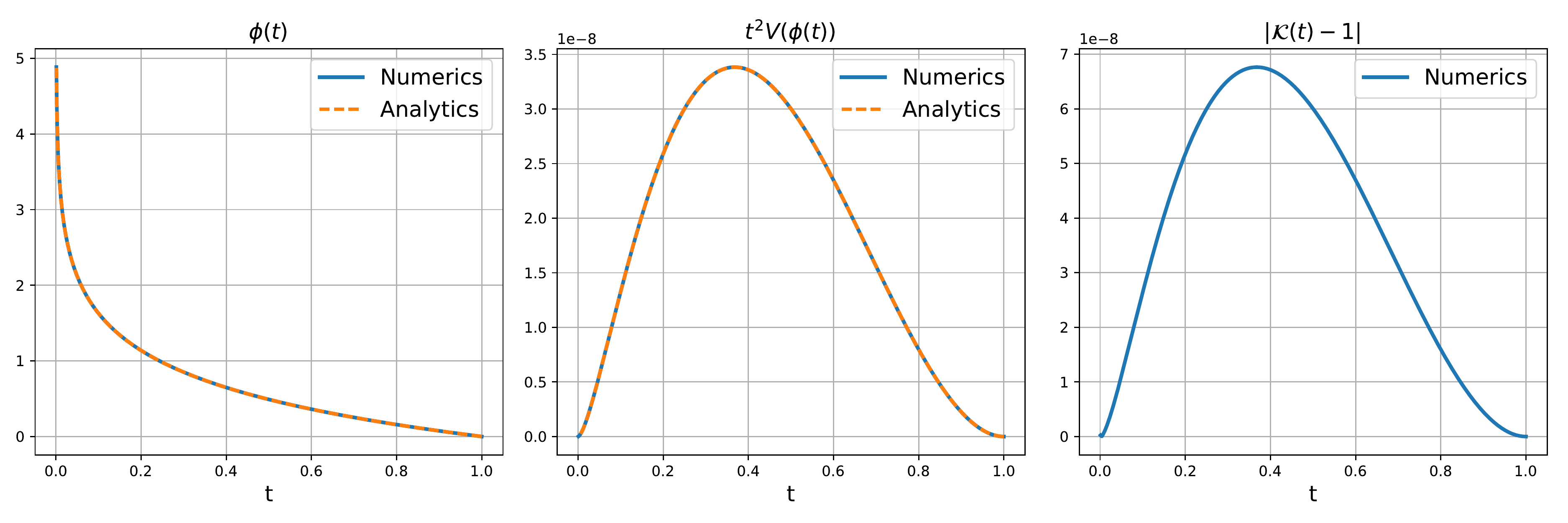}
		\caption{Scalar field solutions with potential as in \Eqref{Eq:CoshPotential}. Here we have set $m=1,\theta=-\pi/3,\psi=5\pi/3$ and $k=10^{-3}$. The left most graph shows the scalar field $\phi$, the centre graph shows $t^2 V(\phi)$, and the right most graph shows the violation of the Kasner constraint $|\mathcal{K}(t)-1|$.}
		\label{fig:kexpansiontests}
	\end{figure}

	In \Figref{fig:kexpansiontests} we test the validity of our approximations for $k=10^{-3},\theta=-\pi/3,\psi=5\pi/3$ and $m=1$. The left plot in \Figref{fig:kexpansiontests} shows the scalar field $\phi$, the centre plot shows the quantity $t^2 V(\phi)$, and the right most plot shows the violation of the Kasner constraint $|\mathcal{K}(t)-1|$. In the first two plots of \Figref{fig:kexpansiontests} we see that the analytical solutions closely agree with the numerical solutions.

	Given all this, one may ask \emph{how do we calculate the asymptotic scalar field strength $A$?} Analytically, we do this by considering the conjugate momentum of the scalar field $\pi_{\phi}:=t\nu$. 
	\begin{align}
	\pi_{\phi} = \mathring{A} + \frac{\mathring{A}}{8}\left( \left( 2 - 3 \mathring{A}^2 \right)(t^2 - 1) + 2t^{2}\left( 3\mathring{A}^2 -2 + 3 \mathring{A}^2 \ln(t) \right)\ln(t) \right)mk^{2} + O(k^4)
	\end{align}
	One now calculates the final value of the scalar field strength $A$ by considering the limit of $\pi_{\phi}$ as $t\rightarrow 0^+$:
	\begin{align}
	A = \lim_{t\rightarrow 0^+}\pi_{\phi}=\mathring{A}+\frac{1}{8}\mathring{A}\left( 3 \mathring{A}^2 -2 \right)m k^2 + O(k^4).
	\label{Eq:Af_Small_k}
	\end{align} 
	\begin{figure}
		\centering
		\includegraphics[width=0.8\linewidth]{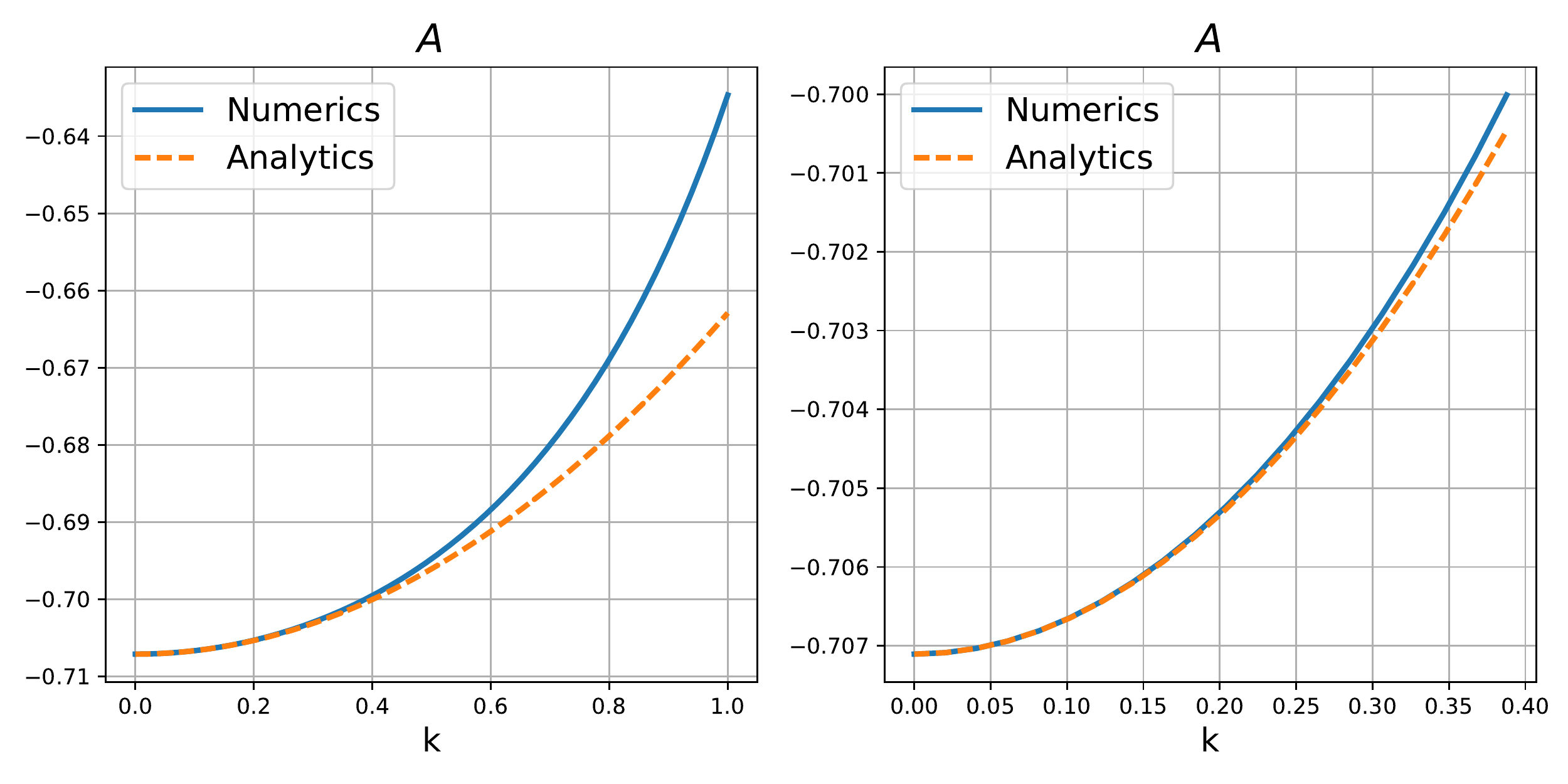}
		\caption{The scalar field strength as a function of $k$ with $m=1,\theta=-\pi/3$ and $\psi=5\pi/3$ so that $\mathring{A}=-1/\sqrt{2}$. The left plot shows $A$ for various values $k\in [0,1]$ and the right plot shows $A$ for $k\in[0,2/5]$. All values of $A$ are numerically calculated at $t=10^{-3}$.}
		\label{fig:aofk}
	\end{figure}
	Numerically, we calculate $A$ by simply taking the final values of $t$ and $\nu$ and multiplying them together. In \Figref{fig:aofk} we show the numerically calculated scalar field strength $A$ as a function of $k$ for $\theta=-\pi/3,\psi =5\pi/3$ and $m=1$ calculated at $t=10^{-3}$. The left plot of \Figref{fig:aofk} shows $A$ for $k\in[0,1]$ and the right for $k\in[0,2/5]$. It is easy to see that \Eqref{Eq:Af_Small_k} gives a reasonably accurate prediction of the scalar field strength when $k\le 2/5$, with accuracy decreasing significantly beyond this point.

	\begin{figure}[t]
		\centering
		\includegraphics[width=1.0\linewidth]{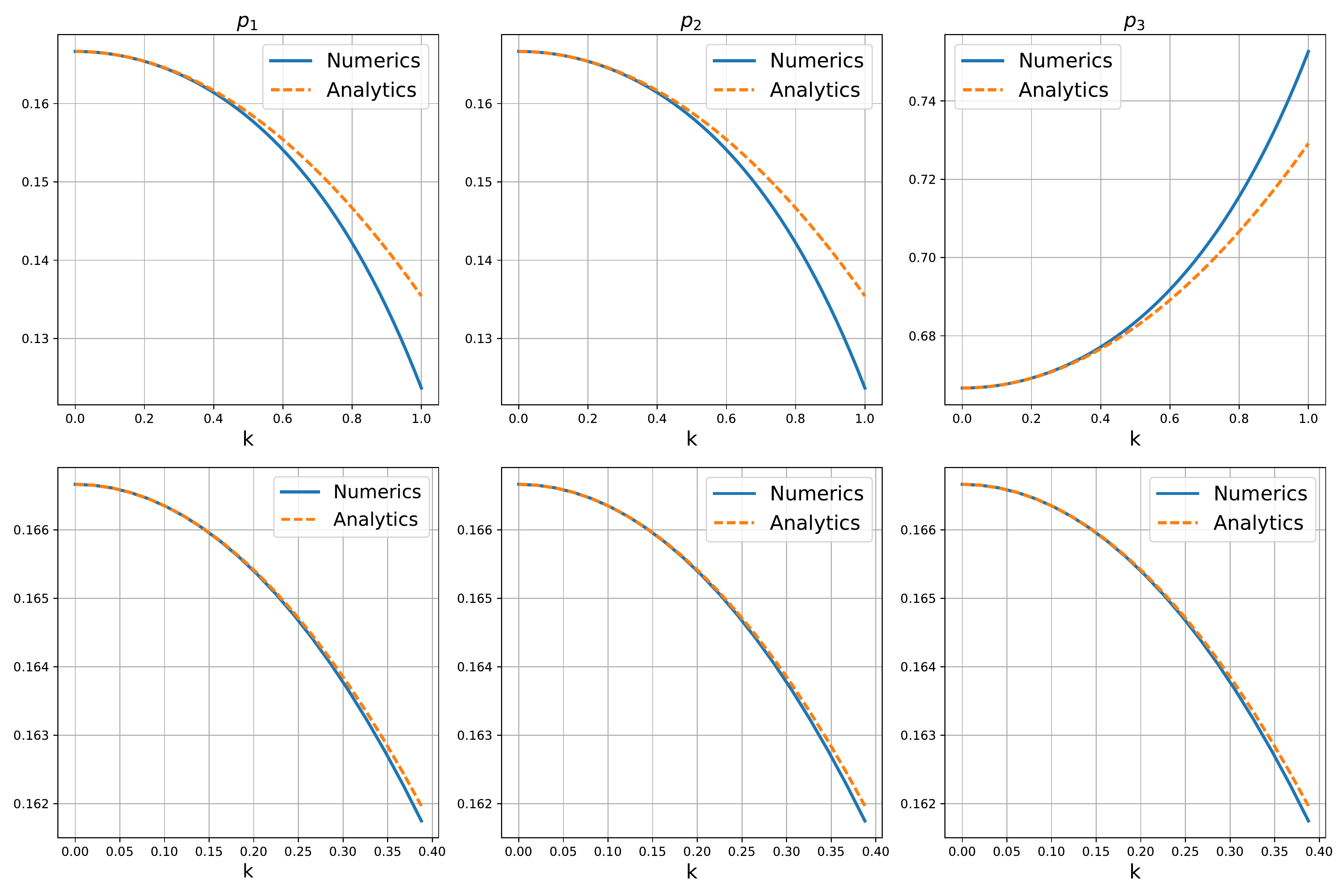}
		\caption{The Kasner exponents as a function of $k$ with $m=1,\theta=-\pi/3$ and $\psi=5\pi/3$ so that $\mathring{A}=-1/\sqrt{2}$. All of the Kasner exponents were numerically calculated at $t=10^{-3}$.}
		\label{fig:pofk}
	\end{figure}

	It is clear from the plots shown in \Figref{fig:kexpansiontests} that the solution (corresponding to $k=10^{-3}$) is asymptotically Kasner. We claim that this is generically true for sufficiently small $k$.

	In order to provide further evidence for this it is useful to now calculate the Kasner exponents. In accordance with \Defref{Def:AsymptoticallyKasner}, this is done by first determining the trace-free part of the second fundamental form ${\chi^a}_{b}$ and then calculating the limits of $t{\chi^a}_{b}$ as $t\rightarrow0^+$. We first calculate ${\chi^a}_b$. Owing to the homogeneity of this problem, we expect ${\chi^a}_b$ to take the form 
	\begin{align}
	{\chi^a}_b = u(t)\text{diag}\left( \frac{1}{3} - \mathring{p}_{1}, \frac{1}{3} - \mathring{p}_2, \frac{1}{3} - \mathring{p}_3 \right),
	\label{Eq:ExCurveAnsatz}
	\end{align}
	where $\mathring{p}_1,\mathring{p}_2,\mathring{p}_3$ are the Kasner exponents of the background solution $(\mathring{\gamma}_{ab},{\mathring \chi}^a{}_b)$ and $u(t)$ is an unknown function determined as a solution of the differential equation
	\begin{align}
	\partial_t u(t) = -\frac{u(t)}{t(1-3t^2 V(\phi(t)))},\quad u(1)=1,
	\label{Eq:u_Equation}
	\end{align}
	which follows from \Eqref{Eq:ADM_Homogenous_A} and \Eqref{Eq:ExCurveAnsatz}. In the special case $k=0$ we have that $V(\phi)=0$, in which case the exact solution of \Eqref{Eq:u_Equation} is $u(t)=1/t$. We therefore expect $u(t)$ (for small $k\neq 0$) to have the asymptotic form
	\begin{align}
	u(t)=\frac{1}{t} + u_{(1)}(t)k + u_{(2)}(t)k^2 + u(t)_{(3)}k^3 + O(k^4), \quad u_{(i)}(1)=0,\quad i=1,\dots,\infty.
	\label{Eq:u_expansion}
	\end{align} 
	We find that if $i$ is odd, then 
	\begin{align}
	\partial_t {u_{(i)}(t)}+\frac{u_{(i)}(t)}{t}=0,\quad u_{(i)}(1)=0.
	\end{align}
	It follows then that $u_{(i)}=0$ if $i$ is odd and as such we have that $u_{(2)}$ is the only non-zero function in the expansion \Eqref{Eq:u_expansion}. By inputting \Eqref{Eq:u_expansion} into \Eqref{Eq:u_Equation}, we find that $u_{(2)}$ must satisfy the differential equation
	\begin{align}
	\partial_t {u_{(2)}(t)}+\frac{u_{(2)}}{t}=-\frac{3}{4}\mathring{A}^2 m \ln(t)^{2},\quad u_{(2)}(1)=0,
	\end{align}
	and is therefore
	\begin{align}
	u_{(2)}=-\frac{3\mathring{A}^{2}m( t^2 - 1 + 2t^{2}( \ln(t)-1 )\ln(t) )}{8t}.
	\end{align}
	Having found an expression for the $u(t)$ (at least in leading order), we can now calculate the time-dependent Kasner exponent $p_{i}(t)$ as (recall from \Sectionref{Sec:Kasner_solutions_with_a_scalar_field} that $p_{i}(t)$ is the ith eigenvalue of $t {K^{a}}_{b}$),
	\begin{align}
	p_{i}(t)=\mathring{p}_{i} - \frac{(3\mathring{p}_i - 1)\mathring{A}^2}{8}\left( t^2- 1 + 2t^2 (\ln(t)-1)\ln(t) \right)mk^2 + O( k^4 ),
	\end{align} 
	which follows from discussions in \Sectionref{Sec:Kasner_solutions_with_a_scalar_field} and \Sectionref{Sec:Asymptotically_Kasner_solutions_in_Bianchi_I_cosmologies}. The final value of the Kasner exponent\footnote{The \emph{number} $p_i=\lim_{t\rightarrow 0^{+}}p_{i}(t)$ should not be mistaken for the \emph{function} $p_i(t)$.} $p_i$ can now be calculated by considering the limit of $p_{i}(t)$ as $t\rightarrow 0^+$:
	\begin{align}
	p_{i}= \lim_{t\rightarrow 0^+}p_{i}(t)=\mathring{p}_{i} + \frac{\mathring{A}^2}{8}\left( 3 \mathring{p}_{i} - 1 \right)mk^{2} + O( k^4 ).
	\label{Eq:pf_Small_k}
	\end{align}

	\begin{figure}
		\centering
		\includegraphics[width=1.0\linewidth]{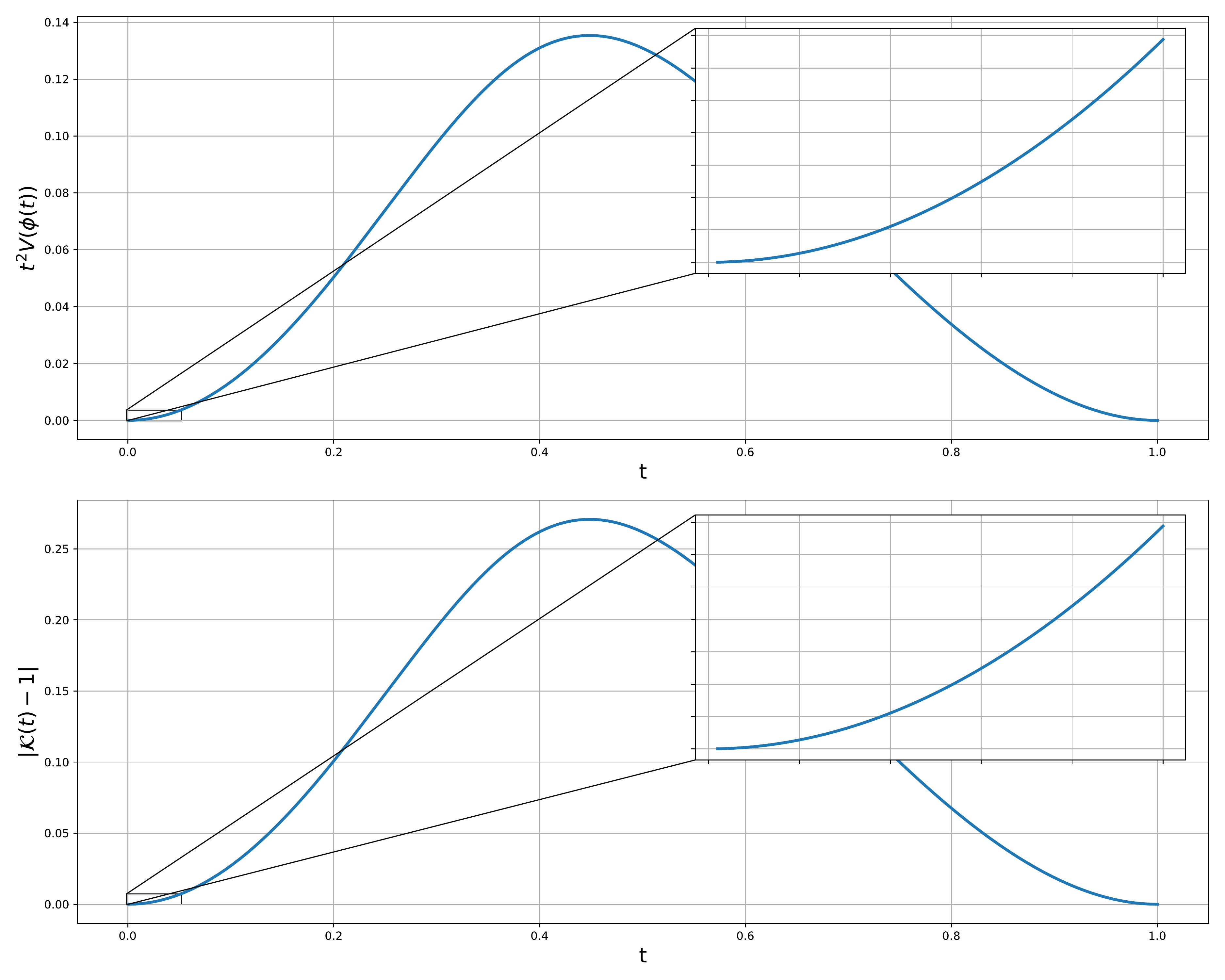}
		\caption{Numerical simulation of an asymptotically stationary (that is also asymptotically Kasner) with $m=10^2, \theta=-\pi/3,\psi=5\pi/3$ and where $k$ is given \Eqref{Eq:Stationary_k_example}. The top plot shows the Kasner condition and the bottom plot shows the violation of the Kasner constraint.}
		\label{fig:kexpansionstationary}
	\end{figure}

	In \Figref{fig:pofk} we show the numerically calculated Kasner exponents for various values of $k$. The left, centre and right columns in \Figref{fig:pofk} show $p_{1},p_{2}$ and $p_{3}$, respectively. The first row shows the Kasner exponents for $k\in[0,1]$ and the second row shows the Kasner exponents for $k\in[0,2/5]$. In \Figref{fig:pofk}, we see that \Eqref{Eq:pf_Small_k} provides a reasonably accurate estimate of the Kasner exponents for $k\in[0,2/5]$.
	
	Using \Eqref{Eq:Af_Small_k} and \Eqref{Eq:pf_Small_k} to numerically calculate the Kasner constraint $\K(0)$ at $t=0$ now gives
	\begin{align}
	\begin{split}
	\mathcal{K}(0)&=\left( \mathring p_{1}^{2} + \mathring p_{2}^{2} + \mathring p_{3}^{2} + \mathring A^{2} \right) + 3\left( ( \mathring A^2 + \mathring p_{1}^{2} + \mathring p_{2}^{2} + \mathring p_{3}^{2} ) - (\mathring  p_1 + \mathring p_2 +  \mathring p_3 ) \right)k^2 + O(k^4).
	\end{split}
	\end{align}
	Since $\mathring p_1, \mathring p_2, \mathring p_3$ and $\mathring A$ are known to satisfy the Kasner relations we immediately find that $\mathcal{K}(0)=1+O(k^4)$ and hence, at least for small $k$, these solutions are asymptotically Kasner.

	\subsubsection{An asymptotically stationary solution}
	\label{SubSec:An_asymptotically_stationary_solution}
	We now search for asymptotically stationary solutions. For this we require that $\nu\rightarrow 0$ as $t\rightarrow 0^+$. In order to find such a solution, we proceed as follows: For fixed $m,\theta$ and $\psi$ we use \Eqref{Eq:Af_Small_k} to obtain an estimate for $k=k_\star$ such that $A(k_\star)=0$. The initial guess is given by the formula 
	\begin{align}
	k_\star =\frac{2\sqrt{2}}{\sqrt{(2-3\mathring{A}{}^2)m}}.
	\label{Eq:AssStation_Guess}
	\end{align}
	We then numerically calculate the scalar field strength $A$ for various values of $k$ in a neighbourhood of $k_\star$. If $A(k)$ changes sign inside of our chosen interval then we apply a bisection method to find the root.  It is interesting to note that one does not expect to find solutions that are asymptotically stationary if both $k$ and $m$ are small. This follows both from \Eqref{Eq:AssStation_Guess} and from \Figref{fig:aofk}. Moreover, we find that one does not expect $k_\star$ to exist if the solution is isotopic (and hence $\mathring{A}^2=2/3$). For the case $m=10^2, \theta=-\pi/3$ and $\psi=5\pi/3$ \Eqref{Eq:AssStation_Guess} gives $k_\star = 2/5$. We numerically find that the scalar field strength $A(k)$ changes sign in the interval $k\in[1/10,7/10]$. Determining the root gives 
	\begin{align}
	k = 0.20346852151752604,\quad A=-2.67\times 10^{-15}.
	\label{Eq:Stationary_k_example}
	\end{align}
	In \Figref{fig:kexpansionstationary} we show the violation of the Kasner constraint as a function of time, where $k$ is given by \Eqref{Eq:Stationary_k_example}. In \Figref{fig:kexpansionstationary} we see that $|\K(t)-1|$ tends to zero as a function of time and so we conclude that this solution is asymptotically Kasner. For each of these simulations, we numerically calculated the scalar field strength $A$ at $t=10^{-3}$. It is interesting to note that in \Figref{fig:kexpansionstationary} we see that $t^2 V(\phi(t))\rightarrow 0$ in the limit $t\rightarrow 0^+$.

	\subsection[Expansions in $\mathring{A}$ and their numerical justifications]{Perturbation expansions in $\mathring{A}$ and their numerical justifications}
	\label{SubSec:Perturbation_expansions_in_A}
	If $mk^2 = O(1)$ then we do not expect the expansions derived in \Sectionref{SubSec:Perturbation_expansions_in_k} to hold as this would violate the ``smallness'' condition of $k$. In this section here we investigate the behaviour of solutions when neither $k$ nor $m$ is small. Instead we suppose that $\mathring{A}\ll 1$. In this case we expect the scalar field $\phi$ to have the asymptotic form
	\begin{align}
	\phi=\mathring{A}\phi_{(1)} + \mathring{A}^2 \phi_{(2)} + \mathring{A}^3 \phi_{(3)} + \mathring{A}^4 \phi_{(4)} +O(\mathring{A}^5).
	\label{Eq:PerturbationExpansion_A}
	\end{align}
	Note here that the quantity $\mathring{A}$ is a natural choice of perturbation parameter as it is bounded below $1$, with $|\mathring{A}|\le\sqrt{2/3}$. Inputting \Eqref{Eq:PerturbationExpansion_A} into \Eqref{Eq:ADM_Homogenous_phi} we find that if $i$ is an even number then the corresponding $\phi_{(i)}$ must satisfy the equation
	\begin{align}
	\partial_{t}^2 {\phi_{(i)}}+\frac{1}{t}\partial_{t} {\phi_{(i)}}+ mk^2 \phi_{(1)}=0,\quad \phi_{(1)}(1)=0, \quad \partial_{t} {\phi_{(1)}}(1)=0.
	\end{align}
	It follows then that, if $i$ is even, we have $\phi_{(i)}=0$ and in particular $\phi_{(2)}=\phi_{(4)}=0$. We therefore find that $\phi_{(1)}$ and $\phi_{(3)}$ are the only non-zero functions in \Eqref{Eq:PerturbationExpansion_A}. It now follows from \Eqref{Eq:ADM_Homogenous_phi} and \Eqref{Eq:Cosh_ID} that $\phi_{(1)}$ is a solution of the differential equation
	\begin{align}
	\partial_{t}^2 {\phi_{(1)}}+\frac{1}{t}\partial_{t} {\phi_{(1)}}+ mk^2 \phi_{(1)}=0,\quad \phi_{(1)}(1)=0, \quad \partial_{t} {\phi_{(1)}}(1)=1,
	\end{align} 
	and is therefore
	\begin{align}
	\phi_{(1)} = \frac{\pi}{2}\left( J_{0}(\sqrt{m}k) Y_{0}( \sqrt{m}k t ) - Y_{0}(\sqrt{m}k) J_{0}( \sqrt{m}k t ) \right),
	\label{Eq:BesselSolution}
	\end{align}
	where $J_{0}()$ and $Y_{0}()$ are Bessel functions of the first and second kind, respectively. Similarly, from \Eqref{Eq:ADM_Homogenous_phi}, we find that $\phi_{(3)}$ must solve the equation
	\begin{align}
	\partial^2_t{\phi_{(3)}}+\frac{1}{t}\partial_t{\phi_{(3)}}+ mk^2 \phi_{(3)}=f(t),\quad \phi_{(3)}(1)=0,\quad \partial_{t}\phi_{(3)}(1)=0,
	\end{align}
	with 
	\begin{align}
	f(t) = \frac{1}{6}k^2 m \left( 18\left( t\partial_{t}\phi_{(1)} \right)^2 + 9 t \left( 2\partial_{t}\phi_{(1)}+t\partial_{t}^2 \phi_{(1)} \right)\phi_{(1)} - k^2\left( 1 + 9 m t^2 \right)\phi_{(1)}^2 \right)\phi_{(1)}.
	\end{align}
	\begin{figure}
		\centering
		\includegraphics[width=1.0\linewidth]{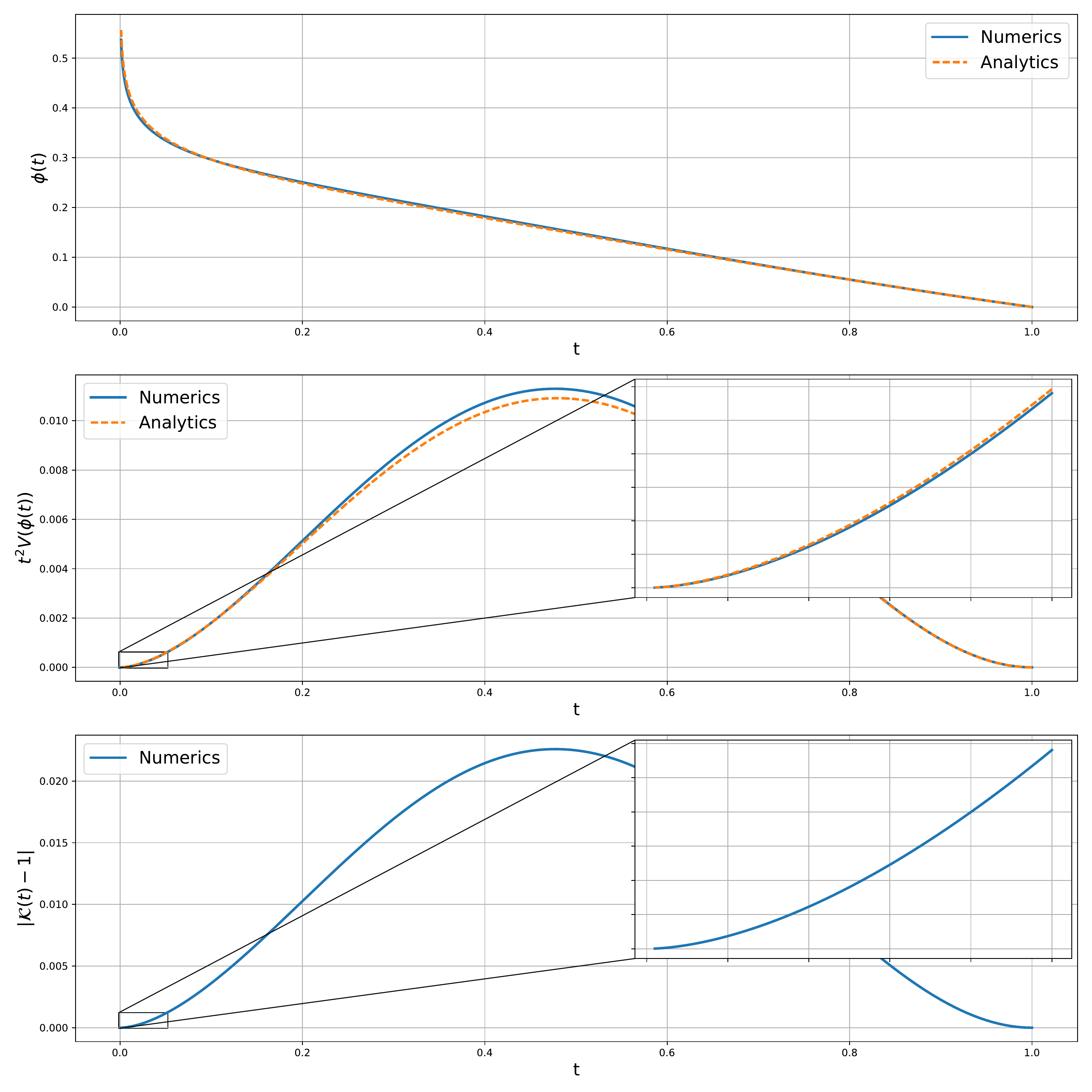}
		\caption{The solutions corresponding to the parameter choices $k=2,m=1,\psi=5\pi/3$ and $\theta=-\pi/10$. In the top, centre and bottom plots we show the scalar field $\phi$, the Kasner condition $t^2 V(\phi)$, and the violation of the Kasner constraint $|\K(t)-1|$, respectively.}
		\label{fig:oscillations02}
	\end{figure}
	\begin{figure}
		\centering
		\includegraphics[width=1.0\linewidth]{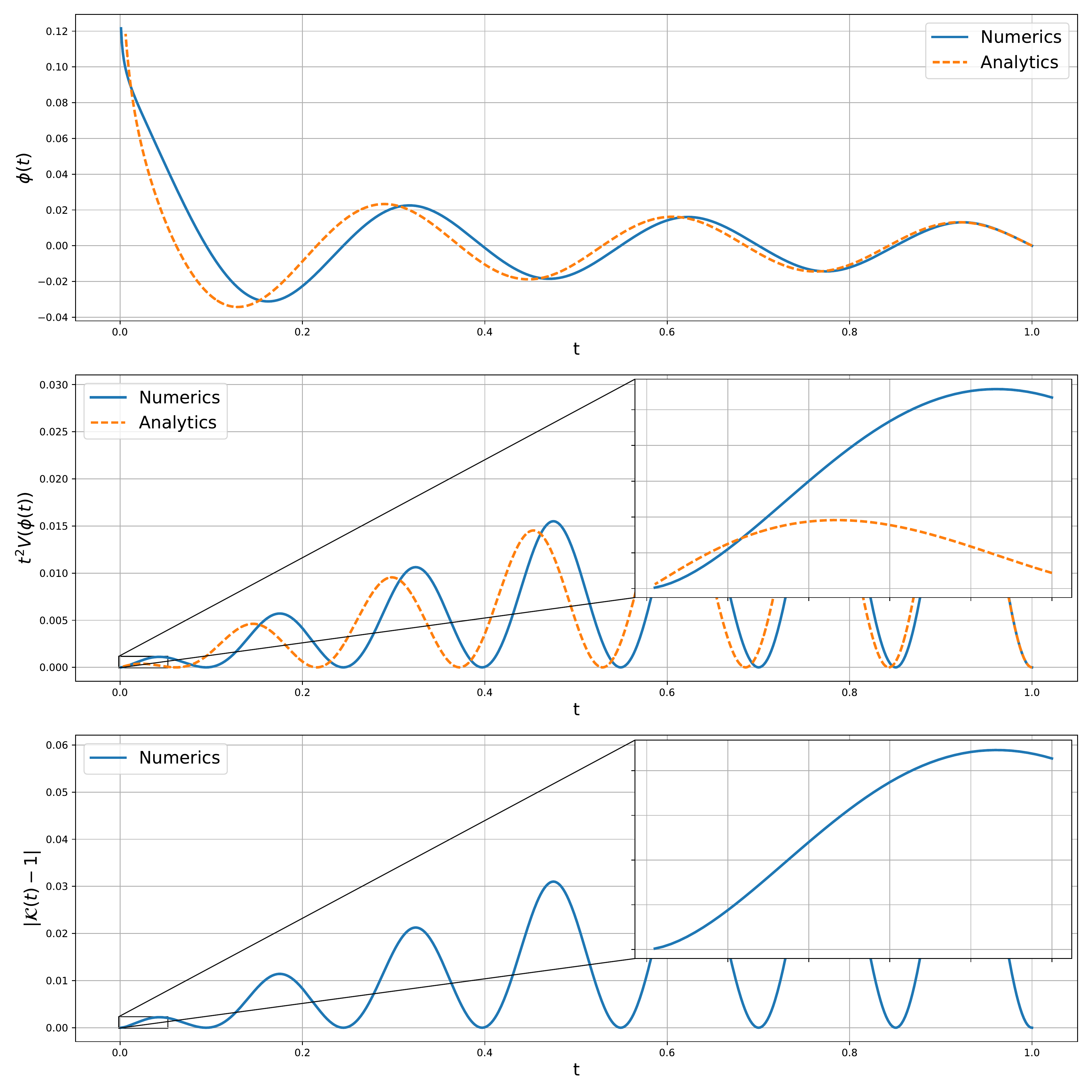}
		\caption{The solutions corresponding to the parameter choices $k=20,m=1,\psi=5\pi/3$ and $\theta=-\pi/10$. In the top, centre and bottom plots we show the scalar field $\phi$, the Kasner condition $t^2 V(\phi)$, and the violation of the Kasner constraint $|\K(t)-1|$, respectively.}
		\label{fig:oscillations20}
	\end{figure}
	The solution is therefore
	\begin{align}
	\phi_{(3)}=Y_{0}( \sqrt{m}k t )\int_{1}^{t}\tau{J_{0}(\sqrt{m}k\tau)f(\tau)}d\tau - J_{0}( \sqrt{m}k t )\int_{1}^{t}\tau{ Y_{0}(\sqrt{m}k\tau ) f(\tau)}d\tau
	\label{Eq:OSC_phi3}
	\end{align}
	We find that we are unable to explicitly integrate \Eqref{Eq:OSC_phi3}. Nevertheless, we can still calculate $\phi_{(3)}$ numerically. To numerically integrate \Eqref{Eq:OSC_phi3} we use the \emph{SciPy} integrator \emph{quad}\footnote{See \url{https://docs.scipy.org/doc/scipy/reference/generated/scipy.integrate.quad.html}.}. We note here that these solutions are only expected to hold only when $mk^2 = O({1})$ and $\mathring{A}\ll 1$. 
	
	We now provide numerical support for the expansions \Eqsref{Eq:PerturbationExpansion_A}--\eqref{Eq:OSC_phi3}. In \Figref{fig:oscillations02} we show the results of our numerical simulation corresponding to the parameter choices $k=2,m=1,\psi=5\pi/3$, and $\theta=-\pi/10$. In the top plot of \Figref{fig:oscillations02} we see that the numerically calculated scalar field $\phi$ closely matches our analytical solution. The middle and bottom plots in \Figref{fig:oscillations02} show the Kasner condition $t^2 V(\phi)$ and the violation of the Kasner constraint $|\K(t)-1|$, respectively. These plots demonstrate that the numerically calculated solution is asymptotically Kasner in the sense of \Defref{Def:AsymptoticallyKasner}. 
	
	In order to test the limitations of our analytical solution, we now consider the case defined by setting $k=20,m=1,\psi=5\pi/3$, and $\theta=-\pi/10$. The results of our numerical tests are shown in \Figref{fig:oscillations20}. The bottom two plots of \Figref{fig:oscillations20} demonstrate that the numerically calculated solutions are asymptotically Kasner,  which is consistent with our analytical predictions. However, our approximations oscillate with a slower frequency than the numerical solutions.

	\section{Conclusion}
	\label{Conclusions}
	In this work we have discussed the asymptotic behaviour of anisotropic space-times that	were constructed as solutions of the Einstein scalar field equations. One of the primary goals of this work was to establish whether or not it is possible to construct ``asymptotically Kasner solutions" with a non-zero potential. We found that the resulting solutions were asymptotically Kasner only if $t^{2}V(\phi(t))\rightarrow 0$ as $t\rightarrow 0^{+}$. Although previous works, such as \cite{PhysRevD.61.023508,KasnerSolutions}, have noted that this is a necessary condition, to the best of our knowledge \Thmref{Result:Kanser} is the first proof that demonstrates it is sufficient.

	For our analytical investigations, we restricted our attention to spatially homogeneous solutions with a strictly monotonic scalar field. This is the key to our proof and although it may seem restrictive at first, we emphasize that the scalar field $\phi$ only needs to be monotonic on some small interval near $t=0$ and hence this result covers all spatially homogeneous solutions. We found that there are three different types of asymptotically Kasner solutions. Namely solutions that are (1) strictly monotonic, (2) eventually monotonic, and (3) asymptotically stationary.

	As with \cite{RodnianskiSpeck:Linear,RodnianskiSpeck:NonLinear}, our analytical treatment relies of the use of CMC coordinates. We found that this gauge choice was not well suited to the investigation of isotropic space-times with a scalar field that is not strictly monotonic and a non-zero potential. In fact, such space-times necessarily lead to singular behaviour at a finite time. We claim that this is a coordinate singularity, however it is unclear if this is the case. Nevertheless we were able to support this claim by calculating three curvature invariants. We found that all three of these curvature invariants remained finite near the singularity.

	By specifying the potential $V(\phi(t))$ as a simple function of \emph{time} (instead of a simple function of $\phi$) we were able to find two new solutions of the Einstein scalar-field equations. On the one hand, we provided an asymptotically Kasner solution with an unbounded potential. On the other hand, we gave a solution that was \emph{not} asymptotically Kasner. For both of these solutions the potential has a simple dependence on time, but a complicated dependence on the scalar field. In future works it would be interesting to further investigate other properties of these solutions, such as stability.

	To extend our investigations we numerically studied the asymptotic behaviour of solutions corresponding to a two parameter \emph{cosh} potential (see \Eqref{Eq:CoshPotential}). This choice of potential allowed us to construct numerical examples of each of the three types of asymptotically Kasner space-times.

	We began by considering the spatially homogeneous setting. Using perturbation expansions, we demonstrated that the resulting solutions are always asymptotically Kasner. In the case when we have $|k\mathring{A}|-2<0$ we found that our perturbation expansions closely matched our numerical simulations. Conversely, we found that if $|k\mathring{A}|-2>0$ then the numerical and analytical results matched only for a short while before the numerical solution \emph{bounced} to a different asymptotically Kasner solution. In the case $|k\mathring{A}|-2=0$ we found that our numerical scheme was not good enough to determine whether or not the solution is asymptotically Kanser. In future works it would be interesting to see if it possible to remedy this. Finally, we investigated spatially homogenous space-times with an eventually monotonic scalar field. Here we found that if $mk^2=O(1)$ then the scalar field exhibited intermediary oscillatory behaviour before becoming eventually monotonic. Through the use of perturbation expansions we were able to show that the scalar field oscillated with frequency $k\sqrt{m}$ (provided $mk^2=O(1)$).

	Although we have restricted our attention to a minimally coupled scalar field, in future works it could be interesting to investigate how our results may change when one instead considers a non-minimally coupled scalar field.

\bibliographystyle{unsrt}
\bibliography{bibfile}
\end{document}